\documentclass[11pt]{article}
\usepackage{lmodern}
\usepackage{xspace} 
\usepackage{amsfonts,amsmath,amssymb, amsthm, mathtools}
\usepackage[colorlinks=true,linkcolor=blue,citecolor=blue,urlcolor=blue]{hyperref}

\usepackage{algorithmicx,algpseudocode,algorithm}
\usepackage[margin=1in]{geometry}
\usepackage{thm-restate}
\usepackage{mleftright} 
\usepackage[shortlabels]{enumitem}

\usepackage[numbers]{natbib}
\usepackage{csquotes}
\usepackage{multirow}
\usepackage{chngpage} 
\usepackage{tikz}
\usepackage{bm}

\newenvironment{proofof}[1]{\begin{proof}[Proof of {#1}]}{\end{proof}}

\newcommand{\eps}{\ensuremath{\varepsilon}\xspace}
\newcommand{\Tester}{\ensuremath{\mathcal{T}}\xspace} 
\newcommand{\property}{\ensuremath{\mathcal{P}}\xspace} 
\newcommand{\eqdef}{\stackrel{\rm def}{=}}
\newcommand{\accept}{\textsf{ACCEPT}\xspace}
\newcommand{\reject}{\textsf{REJECT}\xspace}
\newcommand{\half}{\frac{1}{2}}
\newcommand{\domain}{\ensuremath{[n]}\xspace} 
 
\newcommand{\yes}{\textsf{yes}\xspace}
\newcommand{\no}{\textsf{no}\xspace}
\newcommand{\dyes}{\ensuremath{\cal Y}\xspace}
\newcommand{\dno}{\ensuremath{\cal N}\xspace}

\newcommand{\littleO}[1]{{o\mleft({#1}\mright)}}
\newcommand{\bigO}[1]{{O\mleft({#1}\mright)}}
\newcommand{\bigTheta}[1]{{\Theta\mleft({#1}\mright)}}
\newcommand{\bigOmega}[1]{{\Omega\mleft({#1}\mright)}}
\newcommand{\tildeO}[1]{\tilde{O}\mleft({#1}\mright)}

\providecommand{\poly}{\operatorname*{poly}}

\newcommand{\setOfSuchThat}[2]{ \!\left\{\; #1 \;\colon\; #2\; \right\} } 			

\newcommand{\dtv}{\operatorname{d_{\rm TV}}}
\newcommand{\totalvardist}[2]{{\dtv\mleft({#1, #2}\mright)}}

\newcommand{\proba}{\Pr}
\newcommand{\probaOf}[1]{\proba\!\left[\, #1\, \right]}

\newcommand{\probaDistrOf}[2]{\proba_{#1}\!\left[\, #2\, \right]}

\newcommand{\supp}[1]{\operatorname{supp}\!\left(#1\right)}

\newcommand{\expect}[1]{\mathbb{E}\!\left[#1\right]}

\newcommand{\shortexpect}{\mathbb{E}}

\newcommand{\uniform}{\ensuremath{\mathcal{U}}}
\newcommand{\uniformOn}[1]{\ensuremath{\uniform_{ #1 } }}

\newcommand{\norm}[1]{\lVert#1{\rVert}}
\newcommand{\normone}[1]{{\norm{#1}}_1}
\newcommand{\abs}[1]{\left\lvert #1 \right\rvert}

\newcommand{\clg}[1]{\!\left\lceil #1 \right\rceil}

\newcommand{\ICOND}{{\sf INTCOND}\xspace}
\newcommand{\SAMP}{{\sf SAMP}\xspace}
\newcommand{\COND}{{\sf COND}\xspace}
\newcommand{\PCOND}{{\sf PAIRCOND}\xspace}

\newcommand{\D}{\ensuremath{D}}

\newcommand{\atom}{A}

\newcommand{\newset}{k_i^A}

\newcommand{\iha}[1]{\bm{E_1}(#1)}
\newcommand{\ihb}[1]{\bm{E_2}(#1)}
\newcommand{\ihc}[1]{\bm{E_3}(#1)}

\newcommand{\ie}{i.\,e.} 
 
\newcommand{\eg}{e.\,g.}
\newcommand{\expref}[2]{\texorpdfstring{\hyperref[#2]{#1~\ref{#2}}}{#1~\ref{#2}}} 
\newcommand{\expeqref}[2]{\texorpdfstring{\hyperref[#2]{#1~\eqref{#2}}}{#1~\eqref{#2}}}

\defcitealias{Sublinear:info:open:66}{Sublinear.info}

\newcommand{\bpair}{bucket-pair\xspace}
\newcommand{\bpairs}{bucket-pairs\xspace}

\theoremstyle{plain}
\newtheorem{theorem}{Theorem}[section]
\newtheorem{lemma}[theorem]{Lemma}

\newtheorem{claim}[theorem]{Claim}
\newtheorem{proposition}[theorem]{Proposition}
\newtheorem{fact}[theorem]{Fact}

\theoremstyle{definition}
\newtheorem{definition}[theorem]{Definition}

\title{A Chasm Between Identity and Equivalence Testing with Conditional Queries}
\author {
Jayadev Acharya\thanks{Cornell University. {\tt acharya@cornell.edu}. Supported by a Cornell University Start Up, and NSF CRII-CIF-1657471. Part of this work was done when the author was a postdoctoral researcher at MIT.}
\and
Cl\'ement L.\ Canonne\thanks{Stanford University. {\tt ccanonne@cs.stanford.edu}. Supported by a Motwani Postdoctoral Fellowship. Part of this work was performed while the author was a graduate student at Columbia University, and supported by NSF grants CCF-1115703 and NSF CCF-1319788.}
\and
Gautam Kamath\thanks{Simons Institute for the Theory of Computing. {\tt g@csail.mit.edu}. Supported as a Microsoft Research Fellow, as part of the Simons-Berkeley Research Fellowship program. This work was supported by ONR N00014-12-1-0999, and NSF grants CCF-0953960 (CAREER) and CCF-1101491.}
}
\begin{document}
\maketitle

\begin{abstract}
A recent model for property testing of probability
distributions (Chakraborty et al., ITCS 2013, Canonne et al., SICOMP 2015) 
enables tremendous savings in the sample complexity of testing algorithms,
by allowing them to condition the sampling on subsets of the domain. 
In particular, Canonne, Ron, and Servedio 
(SICOMP~2015)   
showed that, in this setting, testing identity of an unknown distribution $\D$ 
(\ie, whether $\D=\D^\ast$ for an explicitly known $\D^\ast$) can be done
with a \emph{constant} number of queries (\ie, samples), independent
of the support size $n$~--~in contrast to the required $\Omega(\sqrt{n})$
in the standard sampling model. However, it was unclear whether the same
stark contrast exists 
for the case of testing equivalence, 
where \emph{both} distributions are unknown. 
Indeed, while
Canonne et al.   
established a $\poly(\log n)$-query upper bound for equivalence testing,
very recently brought down to 
$\tildeO{\log\log n}$ by Falahatgar et al. 
(COLT~2015),   
whether a dependence on the domain size $n$ is necessary was still open,
and explicitly posed by Fischer at the Bertinoro Workshop on
Sublinear Algorithms~(2014).    
In this
article,   
we answer the question in the affirmative, showing that
any testing algorithm for equivalence must make
$\bigOmega{\sqrt{\log\log n}}$ queries in the conditional sampling model.
Interestingly, this demonstrates
a   
gap between
identity and equivalence testing, absent in the standard sampling model
(where both problems have sampling complexity $n^{\Theta(1)}$).

We also obtain results on the query complexity of uniformity testing and
support-size estimation with conditional samples. In particular, we answer
a question of Chakraborty et al. 
(ITCS 2013)   
showing that \emph{non-adaptive} uniformity testing indeed requires
$\bigOmega{\log n}$
queries in the conditional model. This is an exponential improvement
on their previous lower bound of $\bigOmega{\log\log n}$, and matches
up to polynomial factors their $\poly(\log n)$ upper bound. For the
problem of support-size estimation, we provide both adaptive and
non-adaptive algorithms, with query
complexities   
$\poly(\log\log n)$ and $\poly(\log n)$, 
respectively,   
and complement them with a
qualitatively 
tight lower bound of
$\bigOmega{\log n}$ conditional queries for non-adaptive algorithms.
\end{abstract}
\newpage
\section{Introduction}\label{sec:intro}

\hfill\blockquote{\em{}No, Virginia, there is no constant-query tester.}\footnote{The curious reader is referred to~\cite{wiki:yesvirginia}.}\medskip

Understanding properties and characteristics of an unknown probability distribution is a fundamental problem in statistics, and one that has been thoroughly studied. However, it is only since the work of Goldreich and Ron~\citep{GRexp:00} and Batu et al.~\citep{BFRSW:00} that the problem has been considered through the lens of theoretical computer science, more particularly in the setting of \emph{property testing}. In this framework, an unknown ``huge object''~--~here a probability distribution over a humongous domain~--~can be accessed only by making a few local inspections, and the goal is to decide whether the object satisfies some prespecified property. While most of the literature focuses on the large sample regime and studies error exponents and rates of convergence, more recent testing algorithms look at these problems in the small-sample regime focusing on the probability of errors and sample complexity. (We refer the reader to~\citep{Fischer,Ron:08testlearn,Ron:10FNTTCS,PropertyTestingICS} for an introduction and surveys on the general field of property testing. Moreover, due to the specificities of our model we will interchangeably use the terms sample and query complexity, referring to conditional samples as ``queries.'')

Over the
subsequent  
decade, a flurry of
work explored
this new area, resulting in a better and often complete understanding of a number of questions in property testing of distributions, or distribution testing (see, 
\eg~\citep{GRexp:00, BFFKRW:01,BKR:04,Paninski:08,RS:09,AcharyaDJOP:11, BFRV:11, Rubinfeld:12:Taming, ILR:12, AcharyaDJOPS:12, CDVV:14,ValiantValiant:14} or \citep{clement:survey:distributions} for a survey). In many cases, these culminated in provably sample-optimal
algorithms,    
all of which    
required         
at least an $n^{\Omega(1)}$ dependence on the domain size $n$ in the sample complexity~--~a dependence which, while sublinear, can still be prohibitively large. However, the standard setting of distribution testing, where one only obtains independent samples from an unknown distribution $\D$, does not encompass \emph{all} scenarios one may encounter. In recent years, alternative models have thus been proposed to capture more specific situations~\citep{GMV:06,CFGM:13,CRS:12,LRR:13,CR:14}. Among these is the \emph{conditional oracle model}~\citep{CFGM:13,CRS:12} which will be the focus of our work. In this setting, the testing algorithm is given the ability to sample from conditional distributions: that is, to specify a subset $S$ of the domain and obtain samples from $\D_S$, the distribution induced by $\D$ on $S$ (the formal definition of the model can be found in \expref{Definition}{def:conditional:oracle}). The hope is that allowing a richer set of queries to the unknown underlying distributions might significantly reduce the number of samples the algorithms need, thereby sidestepping the strong lower bounds that hold in the standard sampling model.

\subsection{Motivation for the conditional model} 

A recent trend in testing and learning circumvents these impossibility results by providing a more flexible type of queries than independent samples. One example can be found in the recent paradigm of \emph{active learning}~\citep{Dasgupta:05,ActiveLearning:09} (and its testing counterpart, \emph{active testing}~\citep{balcan2012active}), which modifies and generalizes the usual unsupervised learning paradigm. In this setting, the algorithm is provided with unlabeled examples only, but can then adaptively \emph{request} the label of any of these examples. While not directly comparable to these two frameworks (which are not applicable to the study of probability distributions), the conditional sampling model we shall work on shares some similarity in spirit. Specifically, it provides the algorithm with some additional power, and the ability to perform (some type of) more powerful queries to the object to be learned or tested. 

The setting of conditional sampling is related to that of \emph{group testing}, where the objective is to identify a set of defective individuals among a large population, by querying whether suitably chosen subsets contain at least one defective individual. Group testing has been a field of interest since the 40's and has remained an active area of research since
(see, 
\eg~\citep{Dorfman:43, DH:00, HD:00, ChanJSA:12}). The type of queries allowed in this framework is reminiscent of conditional sampling, where one obtains samples conditioned on a subset. This connection between group testing and conditional sampling is explored in greater detail in~\citep{ACK:15-isit}.

Of course, a crucial aspect of designing and studying these new models of learning and testing is to understand how justified and natural they are, and argue that they do indeed capture natural situations. In the case of distribution testing, the conditional access model does meet these criteria, as discussed in~\citep{CFGM:13} and~\citep{CRS:12}. Namely, besides the purely theoretical aspect of helping understand the key aspects and limitations of the underlying statistical problems, this framework is characteristic of situations which arise in natural and social sciences. At a very high-level, \emph{any} scenario where an experimenter or practitioner is able to restrict the set of possible outcomes of an experiment or poll~--~\eg, in chemistry, where one might control some factor such as the acidity of a solution; or in sociology by performing stratified polling~--~provides this experimenter with the sort of access granted by the conditional model. 

\subsection{Background and previous work}\label{ssec:intro:background}

We focus in this paper on proving lower bounds for testing two extremely natural properties of distributions, namely \emph{equivalence testing} (``are these two distributions identical?'') and \emph{support-size estimation} (``how many different outcomes can actually be observed?''). Along the way, we use some of the techniques we develop to obtain an upper bound on the query complexity of the latter. We state below an informal definition of these two problems, along with closely related ones (uniformity and {identity} testing). Hereafter, ``oracle access'' to a distribution $\D$ over $[n]=\{1,\dots,n\}$  means access to samples generated independently from  $\D$, and ``far'' is with respect to the total variation distance $\totalvardist{\D_1}{\D_2} = \sup_{S\subseteq \domain} (\D_1(S)-\D_2(S))$ between probability distributions. Moreover, as in usual in property testing, in all the problems below we allow the algorithms to be randomized. In the description of the results below, we will often consider the distance parameter $\eps$ as a constant (and focus on the domain size $n$ as the key parameter).
\begin{description}
  \item[Uniformity testing:] granted oracle access to $\D$, decide whether $\D$ equals $\uniformOn{\domain}$ (the uniform distribution on $[n]$) or is far from it;
  \item[Identity testing:] granted oracle access to $\D$ and the full description of a fixed $\D^\ast$, decide whether $\D$ equals $\D^\ast$ or is far from it; 
  \item[Equivalence (closeness) testing:] granted independent oracle accesses to $\D_1$, $\D_2$ (both unknown), decide whether $\D_1$ and $\D_2$ are equal or far from each other. 
  \item[Support-size estimation:] granted oracle access to $\D$, return a multiplicative approximation of the size of the support\footnotemark{} $\supp{\D}=\setOfSuchThat{x}{\D(x) > 0}$. \footnotetext{For this problem, it is typically assumed that all points in the support have probability mass at least $\bigOmega{1}/n$, as without such guarantee it becomes impossible to give any non-trivial estimate (consider for instance a distribution $\D$ such that $\D(i)\propto 1/2^{in}$).}
\end{description}
It is not difficult to see that each of the second and third problems generalizes the previous, and therefore has query complexity at least as big. All of these tasks are known to require sample complexity $n^{\bigOmega{1}}$ in the standard sampling model (\SAMP); yet, as prior work~\citep{CFGM:13,CRS:12} shows, their complexity decreases tremendously when one allows the more flexible type of access to the distribution(s) provided by a conditional sampling oracle (\COND). For the problems of uniformity testing and identity testing, the sample complexity even becomes a constant provided the testing algorithm is allowed to be \emph{adaptive} (\ie, when the next queries it makes can depend on the samples it previously obtained). 

\paragraph{Testing uniformity and identity.} The identity testing question is a generalization of uniformity testing, where $\D^\ast$ is taken to be the uniform distribution over~$\domain$. The complexity of both tasks is well-understood in the sampling model; in particular, it is known that for both uniformity and identity testing $\bigTheta{\sqrt{n}/{\eps^2}}$ samples are necessary and sufficient (see~\citep{GRexp:00,BFRSW:10,Paninski:08,ValiantValiant:14} for the tight bounds on these problems). 

The uniformity testing problem exemplifies the savings granted by conditional sampling -- as Canonne, Ron, and Servedio~\citep{CRS:12} showed, in this setting only $\tildeO{1/\eps^2}$ adaptive queries\footnote{Here and throughout, we use the notation $\tildeO{f}$ to hide polylogarithmic dependencies on the argument, \ie, for expressions of the form $\bigO{f \log^c f}$ (for some absolute constant $c$).}{} are sufficient (and this is optimal, up to logarithmic factors). 
They further prove that identity testing has constant sample complexity as well, namely $\tildeO{1/\eps^4}$~--~very recently improved to a near-optimal $\tildeO{1/\eps^2}$ by Falahatgar et al.~\citep{UCSD:15}. The power of the \COND model is evident from the fact that a task requiring polynomially many samples in the standard model can now be achieved with a number of samples \emph{independent of the domain size $n$}. 

The aforementioned algorithms crucially leverage the ability to make adaptive conditional queries to the probability distributions. Restricting the study to non-adaptive algorithms, Chakraborty et al.~\citep{CFGM:13} describe a $\poly(\log n, 1/\eps)$-query non-adaptive tester for uniformity, showing that even without the full power of conditional queries one can still get an exponential improvement over the standard sampling setting. They also obtain an $\bigOmega{\log\log n}$ lower bound for this problem, and leave open the possibility of improving this lower bound to a logarithmic dependence. We answer this question, establishing in~\expref{Theorem}{theo:uniform:na:lb} that any non-adaptive uniformity tester must perform $\bigOmega{\log n}$ conditional queries.

\paragraph{Testing equivalence.} The \emph{equivalence testing} problem has been extensively studied over the past decade, and its sample complexity is now known to be $\Theta(\max( {n}^{2/3}/{\eps^{4/3}}, \sqrt{n}/\eps^2))$ in the sampling model~\citep{BFRSW:10, Valiant:11, CDVV:14}. 

In the \COND setting, Canonne, Ron, and Servedio showed that equivalence testing is possible with only  $\poly(\log n,1/\eps)$ queries. 
Concurrent to our work, Falahatgar et al.~\citep{UCSD:15} brought this upper bound down to $\tildeO{(\log\log n)/\eps^5}$, a \emph{doubly exponential} improvement over the $n^{\Omega(1)}$ samples needed in the standard sampling model. However, these results still left open the possibility of a constant query complexity -- given that both uniformity and identity testing admit constant-query testers, it is natural to wonder where equivalence testing lies.\footnote{It is worth noting that an $\bigOmega{\log^c n}$ lower bound was known for equivalence testing in a weaker version of the conditional oracle, \PCOND (where the tester's queries are restricted to being either $[n]$ or subsets of size $2$~\citep{CRS:12}).} 
This question was posed by Fischer at the Bertinoro Workshop on Sublinear Algorithms 2014~\citepalias[Problem 66]{Sublinear:info:open:66}. We make decisive progress in answering it, ruling out the possibility of any constant-query tester for equivalence. Along with the upper bound of Falahatgar et al.~\citep{UCSD:15}, our results nearly settle the dependence on the domain size, showing that $(\log\log n)^{\Theta(1)}$ samples are both necessary and sufficient.

\paragraph{Support-size estimation.} Raskhodnikova et al.~\citep{RRSS:09} showed that obtaining \emph{additive} estimates of the support size requires sample complexity almost linear in $n$. Subsequent work by Valiant and Valiant~\citep{ValiantValiant:11,ValiantValiant:10lb} settles the question, establishing that $\bigTheta{n/\log n}$ samples are both necessary and sufficient. Note that the proof of the Valiants' lower bound translates to multiplicative approximations as well, as it hinges on the hardness of distinguishing a distribution with support $s\leq n$ from a distribution with support $s+\eps n\geq(1+\eps)s$. 
To the best of our knowledge, the question of getting a multiplicative-factor estimate of the support size of a distribution given conditional sampling access has not been previously considered. 
We provide upper and lower bounds for both the adaptive and non-adaptive versions of this problem.

\subsection{Our results}\label{ssec:intro:results}

We make significant progress in each of the problems introduced in the previous section, yielding a better understanding of their 
query complexities. We prove four results pertaining to the sample complexity of equivalence testing, support-size estimation, and uniformity testing in the \COND framework.

Our main result
gives a lower bound on  
the sample complexity of testing equivalence with adaptive queries under the \COND model, resolving in the negative the question of whether constant-query complexity was achievable \citepalias[Problem 66]{Sublinear:info:open:66}. 
\begin{restatable}[Testing Equivalence]{theorem}{thmmainlbequiv}\label{theo:main:lb:equiv}
Any \emph{adaptive} algorithm which, given $\COND$ access to unknown distributions $\D_1,\D_2$ on $\domain$, distinguishes with probability at least $2/3$ between \textsf{(a)} $\D_1 = \D_2$ and \textsf{(b)} $\totalvardist{\D_1}{\D_2} \geq \frac14$, must have query complexity $\bigOmega{\sqrt{\log\log n}}$.
\end{restatable}
\noindent Combined with the recent $\tildeO{\log\log n}$ upper bound of Falahatgar et al.~\citep{UCSD:15}, this almost settles the sample complexity of this question. 
Furthermore, as the related task of identity testing \emph{can} be performed with a constant number of queries in the conditional sampling model, this demonstrates an intriguing
 and intrinsic 
difference between the two problems.
Our result also shows an interesting contrast with    
the usual sampling model, where
both identity and equivalence testing have polynomial sample complexity.
\medskip
Next, we establish a logarithmic lower bound on \emph{non-adaptive} support-size estimation, for any (large enough) constant factor. 
  This improves on the result of Chakraborty et al.~\citep{CFGM:13}, which gave a doubly logarithmic lower bound for constant factor support-size estimation.
\begin{restatable}[Non-Adaptive Support-Size Estimation]{theorem}{thmsuppnalb}\label{theo:support:na:lb:beta}
  Any non-adaptive algorithm which, given $\COND$ access to an unknown distribution $\D$ on $\domain$, estimates the size of its support up to a factor $\gamma \geq \sqrt{2}$ must have query complexity $\bigOmega{\frac{\log n}{\log^2 \gamma}}$.
\end{restatable}

\noindent Moreover, the approach used to prove this theorem also implies an analogous lower bound on \emph{non-adaptive} uniformity testing in the conditional model, answering a conjecture of Chakraborty et al.~\citep{CFGM:13}.
\begin{restatable}[Non-Adaptive Uniformity Testing]{theorem}{thmunifnalb}\label{theo:uniform:na:lb}
  Any non-adaptive algorithm which, given $\COND$ access to an unknown distribution $\D$ on $\domain$ and parameter $\eps\in(0,1/4)$, distinguishes with probability at least $2/3$ between \textsf{(a)}~$\D = \uniformOn{\domain}$ and \textsf{(b)}~$\dtv({\D},\uniformOn{\domain}) \geq \eps$, must have query complexity $\bigOmega{(\log n)/\eps}$.
\end{restatable}
\noindent We note that these results complement $\poly\!\log(n)$-query upper bounds on non-adaptive support-size estimation and uniformity testing, the former of which we sketch in this paper, and the latter obtained by Chakraborty et al.~\citep{CFGM:13}. 
This shows that both of these problems have query complexity $\log^{\Theta(1)} n$ in the non-adaptive case.\medskip

Finally, we conclude with an upper bound for \emph{adaptive} support-size estimation.
Specifically, we provide a $\tildeO{\log \log n}$-query algorithm for support-size estimation.
This shows that the question becomes \emph{double exponentially} easier when conditional samples are allowed.

\begin{restatable}[Adaptive Support-Size Estimation]{theorem}{thmsuppsize}\label{theo:supp:size}
Let $\tau > 0$ be any constant. There exists an adaptive algorithm which, given $\COND$ access to an unknown distribution $\D$ on $\domain$ (guaranteed to have probability mass at least $\tau/n$ on every element of its support) and accuracy parameter $\eps\in(0,1)$, makes $\tildeO{(\log\log n)/\eps^3}$ queries to the oracle\footnote{We remark that the constant in the $\tilde{O}$ depends polynomially on $1/\tau$.} and returns a value $\tilde{\omega}$ such that the following holds. With probability at least $2/3$, $\tilde{\omega}\in[\frac{1}{1+\eps}\cdot\omega, (1+\eps)\cdot\omega]$, where $\omega=\abs{\supp{\D}}$.
\end{restatable}

\newcommand{\pb}[2]{\parbox[c][][c]{#1}{\strut#2\strut}}
\begin{center}
\begin{table}[ht]\centering \renewcommand{\arraystretch}{1.25}
    \begin{adjustwidth}{-.5in}{-.5in}\centering
      \begin{tabular}{| l | c | c |}
      \hline
      \pb{0.25\textwidth}{\centerline{\bf Problem}} & {\bf \COND model} & {\bf Standard model} \\ \hline
      \multirow{2}{*}{\pb{0.30\textwidth}{\centering Testing equivalence}} & $\tildeO{ \frac{\log \log n}{\eps^5} }$ \cite{UCSD:15} & 
        \multirow{2}{*}{$\bigTheta{
            \max\left(\frac{n^{2/3}}{\eps^{4/3}},
              \frac{n^{1/2}}{\eps^2}\right) }$ \cite{CDVV:14} }\\
      & $\bigOmega{ \sqrt{\log \log n} }$ (\expref{Theorem}{theo:main:lb:equiv}) & \\ \hline
      \multirow{2}{*}{\pb{0.30\textwidth}{\centering Estimating support size (adaptive)}} 
      & $\tildeO{ \frac{\log \log n}{\eps^3} }$ (\expref{Theorem}{theo:supp:size}) & \multirow{4}{*}{$\bigTheta{ \frac{n}{\log n} }$ \cite{ValiantValiant:10lb} }\\
      & $\bigOmega{ \sqrt{\log \log n} }$ \cite{CFGM:13} $(\dagger)$ & \\ \cline{1-2} 
      \multirow{2}{*}{\pb{0.30\textwidth}{\centering Estimating support size (non-adaptive)}}  
      & $\bigO{\poly( \log n, 1/\eps )}$ (\expref{Section}{ssec:nonadaptive:support:size:ub}) &  \\
      & $\bigOmega{ \log n }$ (\expref{Theorem}{theo:support:na:lb:beta}) & \\ \hline
      \multirow{2}{*}{\pb{0.30\textwidth}{\centering Testing uniformity (non-adaptive)}} & $\tildeO{\frac{\log^5 n}{\eps^6} }$ \cite{CFGM:13}  
       & \multirow{2}{*}{$\bigTheta{ \frac{\sqrt{n}}{\eps^2} }$ \cite{Paninski:08}}\\
       & $\bigOmega{ \frac{\log n}{\eps} }$ (\expref{Theorem}{theo:uniform:na:lb}) & \\
      \hline
      \end{tabular}
    \end{adjustwidth}
    \caption{\label{fig:ballerass:table}Summary of results. Note that the lower bound $(\dagger)$ can also be easily derived from our lower bound on testing equivalence.}
\end{table}
\end{center}

\subsubsection{Relation to the Ron-Tsur model}
  Recent work of Ron and Tsur~\citep{RonTsur:14} studies a model which is slightly different~--~and more favorable to the algorithm~--~than ours.
  In their setting, the algorithm still performs queries consisting of a subset of the domain, as in our case.
  However, the algorithm is also given the promise that the distribution is uniform on a subset of the domain, and whenever a query set contains $0$ probability mass the oracle explicitly indicates this is the case. 
  Their paper provides a number of results for support-size estimation in this model.

  We point out two connections between our work and theirs.
  First, our $\bigOmega{\log n}$ lower bound for non-adaptive support-size estimation (\expref{Theorem}{theo:support:na:lb:beta}) holds in the model of Ron and Tsur. 
  Although lower bounds in the conditional sampling setting do not apply directly to their model, our construction and analysis do carry over, and provide a nearly tight answer to a question left unanswered in their paper.
  Also, our $\tildeO{\log \log n}$-query algorithm for adaptive support-size estimation (\expref{Theorem}{theo:supp:size}) can be seen as generalizing their result to the weaker conditional sampling model (most significantly, when we are not given the promise that the distribution be uniform).

\subsection{Techniques and proof ideas}\label{ssec:intro:techniques}

\paragraph{Lower bound on adaptive equivalence testing.} In order to prove our main lower bound,~\expref{Theorem}{theo:main:lb:equiv}, we have to deal with one main conceptual issue: \emph{adaptivity}. 
While the standard sampling model does not, by definition, allow any choice on what the next query to the oracle should be, this is no longer the case for \COND algorithms.
Quantifying the power that this grants an algorithm makes things much more difficult. To handle this point, we follow the approach of Chakraborty et al. \citep{CFGM:13} and focus on a restricted class of algorithms they introduce, called ``core adaptive testers'' (see \expref{Section}{ssec:prelim:coretesters} for a formal definition). They show that this class of testers is equivalent to general algorithms for the purpose of testing a broad class of properties, namely those which are invariant to any permutation of the domain. Using this characterization, it remains for us to show that none of these structurally much simpler core testers can distinguish whether they are given conditional access to \textsf{(a)} a pair of random identical distributions $(\D_1,\D_1)$, or \textsf{(b)} two distributions $(\D_1,\D_2)$ drawn according to a similar process, which are far apart. 

At a high level, our lower bound works by designing instances where the property can be tested if and only if the support size is known to the algorithm.
Our construction randomizes the support size by embedding the instance into a polynomially larger domain.
Since the algorithm is only allowed a small number of queries, Yao's Minimax Principle allows us to argue that, with high probability, a deterministic algorithm is unable to ``guess'' the support size.
This separates queries into several cases.
First, in a sense we make precise, it is somehow ``predictable'' whether or not a query will return an element previously observed.
If this occurs, it is similarly predictable \emph{which} element the query will return.
On the other hand, if a fresh element is observed, the query set is either ``too small'' or ``too large.''
In the former case, the query will entirely miss the support, and the sampling process is identical for both types of instance.
In the latter case, the query will hit a large portion of the support, and the amount of information gleaned from a single sample is minimal.

At a lower level, this process itself is reminiscent of the ``hard'' instances underlying the lower bound of Canonne, Ron, and Servedio~\citep{CRS:12} for testing identity (with a \PCOND oracle), with one pivotal twist. As in their work, both $\D_1$ and $\D_2$ are uniform within each of $\omega(1)$ ``buckets'' whose size grows exponentially and are grouped into ``\bpairs.'' Then, $\D_2$ is obtained from $\D_1$ by internally redistributing the probability mass of each pair of buckets, so that the total mass of each pair is preserved but each particular bucket has mass going up or down by a constant factor (see \expref{Section}{ssec:main:lb:construction} for details of the construction). However, we now add a final step, where in both $\D_1$ and $\D_2$ the resulting distribution's support is \emph{scaled by a random factor}, effectively reducing it to a (randomly) negligible fraction of the domain. Intuitively, this last modification has the role of ``blinding'' the testing algorithm. We argue that unless its queries are on sets whose size somehow match (in a sense formalized in \expref{Section}{ssec:main:lb:analysis}) this random size of the support, the sequences of samples it will obtain under $\D_1$ and $\D_2$ are almost identically distributed. The above discussion crucially hides many significant aspects and technical difficulties which we address in \expref{Section}{sec:lb-equiv}. Moreover, we observe that the lower bound we obtain seems to be optimal with regard to our proofs techniques (specifically, to the decision tree approach), and not an artifact of our lower bound instances. Namely, there appear to be conceptual barriers to strengthening our result, which would require new ideas.

\paragraph{Lower bound on non-adaptive support-size estimation.} 
Turning to the (non-adaptive) lower bound of \expref{Theorem}{theo:support:na:lb:beta}, we define two families of distributions $\mathcal{\D}_1$ and $\mathcal{\D}_2$, where an instance is either a draw $(\D_1,\D_2)$ from $\mathcal{\D}_1\times\mathcal{\D}_2$, or simply $(\D_1,\D_1)$. Any distribution in $\mathcal{\D}_2$ has support size $\gamma$ times that of its corresponding distribution in $\mathcal{\D}_1$. 
Yet, we argue that no non-adaptive \emph{deterministic} tester making too few queries can distinguish between these two cases, as the tuple of samples it will obtain from $\D_1$ or (the corresponding) $\D_2$ is almost identically distributed (where the randomness is over the choice of the instance itself). To show this last point, we analyze separately the case of ``small'' queries (conditioning on sets which turn out to be much smaller than the actual support size, and thus with high probability will not even intersect it) and the ``large'' ones (where the query set $S$ is so big compared to the support $T$ that a uniform sample from $S\cap T$ is essentially indistinguishable from a uniform sample from $S$). We conclude the proof by invoking Yao's Principle, carrying the lower bound back to the setting of non-adaptive \emph{randomized} testers.

Interestingly, this argument essentially gives us~\expref{Theorem}{theo:uniform:na:lb} ``for free.'' Indeed, the big-query-set case above is handled by proving that the distribution of samples returned on those queries is indistinguishable, both for $\mathcal{\D}_1$ and $\mathcal{\D}_2$, from samples obtained from the \emph{actual} uniform distribution. Considering again the small-query-set case separately, this allows us to argue that a random distribution from (say) $\mathcal{\D}_1$ is indistinguishable from uniform.

\paragraph{Upper bound on support-size estimation.} Our algorithm for estimating the support size within a constant factor (\expref{Theorem}{theo:supp:size}) is simple in spirit, and follows a guess-and-check strategy. In more detail, it first obtains a ``reference point'' \emph{outside} the support, to check whether subsequent samples it may consider belong to the support. Then, it attempts to find a \emph{rough upper bound} on the size of the support, of the form $2^{2^j}$ (so that only $\log\log n$ many options have to be considered); by using its reference point to check if a uniform random subset of this size contains, as it should, at least one point from the support.
Once such an upper bound has been obtained using this double-exponential strategy, a refined bound is then obtained via a binary search on the new range of values for the exponent, $\{2^{j-1},\dots, 2^{j}\}$. Not surprisingly, our algorithm draws on similar ideas as in~\citep{RonTsur:14,Stock:85}, with some additional machinery to supplement the differences in the models.
Interestingly, as a side-effect, this upper bound shows our analysis of \expref{Theorem}{theo:main:lb:equiv} to be tight up to a quadratic improvement. Indeed, the lower bound construction we consider (see \expref{Section}{ssec:main:lb:construction}) can be easily ``defeated'' if an estimate of the support size is known, and therefore cannot yield better than a $\bigOmega{\log\log n}$ lower bound.
Similarly, this further shows that the adaptive lower bound for support-size estimation of Chakraborty et al.~\citep{CFGM:13} is also tight up to a quadratic improvement.

\paragraph{Organization.} The rest of the paper describes details and proofs of the results mentioned in the above discussion.
 In \expref{Section}{sec:prelim}, we introduce the necessary definitions and some of the tools we shall use. 
 \expref{Section}{sec:lb-equiv} covers our main result on adaptive equivalence testing, \expref{Theorem}{theo:main:lb:equiv}.
 In~\expref{Section}{sec:lb-uniform} we prove our lower bounds for support-size estimation and uniformity testing, and~\expref{Section}{sec:ub-supp-size} details our upper bounds for support-size estimation. 
 The corresponding sections may be read independently.

\section{Preliminaries}\label{sec:prelim}

\subsection{Notation and sampling models}
All throughout this paper, we denote by $[n]$ the set $\{1,\dots,n\}$, and by $\log$ the logarithm in base $2$. A probability distribution over a (countable) domain $\domain$ is a non-negative function $\D\colon\domain\to[0,1]$ such that $\sum_{x\in\domain} \D(x) = 1$. We denote by $\uniformOn{S}$ the uniform distribution on a set $S$. Given a distribution $\D$ over $\domain$ and a set $S\subseteq\domain$, we write $\D(S)$ for the total probability mass $\sum_{x\in S} \D(x)$ assigned to $S$ by $\D$. Finally, for $S \subseteq \domain$ such that $\D(S)>0$, we denote by $\D_S$ the conditional distribution of $\D$ restricted to $S$, that is $\D_S(x) = \D(x)/\D(S)$ for $x \in S$ and $\D_S(x)=0$ otherwise.\medskip

As is usual in distribution testing, in this work the distance between two distributions $\D_1, \D_2$ on $\domain$ will be the \emph{total variation distance}.
\begin{equation}\label{def:distance:tv}
\totalvardist{\D_1}{\D_2 } \eqdef \frac{1}{2} \normone{\D_1 - \D_2} = \frac{1}{2} \sum_{x \in \domain}\abs{\D_1(i)-\D_2(i)} = \max_{S \subseteq \domain} (\D_1(S)-\D_2(S))
\end{equation}
which takes value in $[0,1]$.\medskip

In this work, we focus on the setting of \emph{conditional access} to the distribution, as introduced and studied in \cite{CFGM:13,CRS:12}. We reproduce below the corresponding definition of a conditional oracle, henceforth referred to as $\COND$.

\begin{definition}[Conditional access model]\label{def:conditional:oracle}
Fix a distribution $\D$ over $\domain$.  A \emph{\COND oracle for $\D$}, denoted
$\COND_\D$, is defined as follows:
 the oracle takes as input a \emph{query set}
 $S \subseteq \domain$, chosen by the algorithm,  that has $\D(S) > 0$. The oracle returns an element $i \in S$, where
 the probability that element $i$ is returned is $\D_S(i) = \D(i)/\D(S),$
 independently of all previous calls to the oracle.
\end{definition}
Note that as described above the behavior of $\COND_\D(S)$ is undefined if $\D(S)=0$, \ie, the set $S$ has zero probability under $\D$.  Various definitional choices could be made to deal with this. These choices do not make
a  
significant difference in most situations, as
adaptive\footnotetext{Recall that a non-adaptive tester is an algorithm whose queries do not depend on the answers obtained from previous ones, but only on its internal randomness. Equivalently, it is a tester that can commit ``upfront'' to all the queries it will make to the oracle.} algorithms can
include in their next queries a sample previously obtained;
 while our lower bounds can be thought of as putting exponentially small probability mass of elements outside the support. For this reason, and for convenience, we shall hereafter assume, following Chakraborty et al., that the oracle returns in this case a sample uniformly distributed in $S$.
Furthermore, as in~\cite{CRS:12,CFGM:13} we do not take
the \emph{complexity} of specifying the set $S$ to the oracle into account,
and indeed allow arbitrary sets as queries.\footnote{We further observe that,
besides the general $\COND_\D$ oracle which allows these arbitrary query
sets, the authors of~\cite{CRS:12}   
introduce two weaker variants of the conditional model
(the ``pair-cond'' $\PCOND_\D$ and ``interval-cond'' $\ICOND_\D$ oracles)
which restrict algorithms to ``simple'' queries.}
 
 \medskip

Finally, recall that a \emph{property} $\property$ of distributions over $\domain$ is a set consisting of all distributions that have the property. The distance from $\D$ to a property $\property$, denoted $\totalvardist{\D}{\property}$, is then defined as $\inf_{\D^\prime \in \property} \totalvardist{\D}{\property}$.
We use the standard definition of testing algorithms for properties of distributions over $\domain$, tailored for the setting of conditional access to an unknown distribution. 
\begin{definition}[Property tester]\label{def:testing}
  Let $\property$ be a property of distributions over $\domain$. A \emph{$t$-query $\COND$ testing algorithm for $\property$} is a randomized algorithm $\Tester$ which takes as input $n$, $\eps\in(0,1]$, as well as access to $\COND_\D$.  After making at most $t(\eps,n)$ calls to the oracle, \Tester either returns \accept or \reject, such that the following holds:
  \begin{itemize}
  \item if $\D \in \property$, \Tester returns \accept with probability at least $2/3$;
  \item if $\totalvardist{\D}{\property} \geq \eps$, \Tester returns \reject with probability at least $2/3$.
  \end{itemize}
\end{definition}

We observe that the above definitions can be straightforwardly extended to the more general setting of \emph{pairs} of distributions, where given independent access to two oracles $\COND_{\D_1}$,  $\COND_{\D_2}$ the goal is to test whether $(\D_1,\D_2)$ satisfies a property (now a set of pairs of distributions). This will be the case in \expref{Section}{sec:lb-equiv}, where we will consider equivalence testing, that is the property $\property_{\rm{}eq}=\setOfSuchThat{(\D_1,\D_2)}{\D_1=\D_2}$.

\subsection{Adaptive Core Testers}\label{ssec:prelim:coretesters}

In order to deal with adaptivity in our lower bounds, we will use ideas introduced by Chakraborty et al.~\cite{CFGM:13}. These ideas, for the case of \emph{label-invariant} properties\footnote{Recall that a property is label-invariant (or \emph{symmetric}) if it is closed under relabeling of the elements of the support. More precisely, a property of distributions (resp. pairs of distributions) \property is label-invariant if for any distribution $\D\in\property$ (resp. $(\D_1,\D_2)\in\property$) and permutation $\sigma$ of \domain, one has $\D\circ\sigma\in\property$ (resp. $(\D_1\circ\sigma,\D_2\circ\sigma)\in\property$).} allow one to narrow down the range of possible testers and focus on a restricted class of such algorithms called \emph{adaptive core testers}. These core testers do not have access to the full information of the samples they draw, but instead only get to see the relations (inclusions, equalities) between the queries they make and the samples they get. Yet, Chakraborty et al.~\cite{CFGM:13} show that any tester for a label-invariant property can be converted into a core tester with same query complexity; thus, it is enough to prove lower bounds against this -- seemingly -- weaker class of algorithms.

We here rephrase the definitions of a core tester and  the view they have of the interaction with the oracle (the \emph{configuration} of the samples), tailored to our setting.

\begin{definition}[Atoms and partitions]\label{def:atoms}
  Given a family $\mathcal{A}=(A_1,\dots,A_t)\subseteq[n]^t$, the \emph{atoms} generated by $\mathcal{A}$ are the (at most) $2^t$ distinct sets of the form $\bigcap_{r=1}^t C_r$, where $C_r\in\{A_r, [n]\setminus A_r\}$. The family of all
atoms, denoted $\operatorname{At}(\mathcal{A})$, is the \emph{partition} generated by $\mathcal{A}$.
\end{definition}

This definition essentially captures ``all sets
(besides the $A_i$)  
about which something can be
learned   
from querying the oracle on the sets of $\mathcal{A}$.''
Now, given such a sequence of queries $\mathcal{A}=(A_1,\dots,A_t)$ and pairs of samples $\mathbf{s}=((s^{(1)}_1,s^{(2)}_1),\dots,(s^{(1)}_t,s^{(2)}_t))\in A_1^2\times\dots\times A_t^2$, we would like to summarize ``all the label-invariant information available to an algorithm that obtains $((s^{(1)}_1,s^{(2)}_1),\dots,(s^{(1)}_t,s^{(2)}_t))$ upon querying $A_1,\dots,A_t$ for $\D_1$ and $\D_2$.'' This calls for the following definition.
\begin{definition}[$t$-configuration]\label{def:configuration}
Given $\mathcal{A}=(A_1,\dots,A_t)$ and $\mathbf{s}=((s^{(1)}_j,s^{(2)}_j))_{1\leq j\leq t}$ as above, the \emph{$t$-configuration of $\mathbf{s}$} consists of the $6t^2$ bits indicating, for all $1\leq i,j\leq t$, whether
\begin{itemize}
  \item $s^{(k)}_i = s^{(\ell)}_j$, for $k,\ell\in\{1,2\}$; and \hfill (relations between samples)
  \item $s^{(k)}_i \in A_j$, for $k\in\{1,2\}$. \hfill (relations between samples and query sets)
\end{itemize}
In other terms, it summarizes which is the unique atom $S_i\in \operatorname{At}(\mathcal{A})$ that contains $s^{(k)}_i$, and what collisions between samples have been observed.
\end{definition}

As aforementioned, the key idea is to argue that, without loss of generality, one can restrict one's attention to algorithms that only have access to $t$-configurations, and generate their queries in a specific (albeit adaptive) fashion.
\begin{definition}[Core adaptive tester]\label{def:core:tester}
  A \emph{core adaptive distribution tester} for pairs of distributions is an algorithm \Tester that acts as follows.
  \begin{itemize}
    \item In the $i$-th phase, based only on its own internal randomness and the configuration of the previous queries $A_1,\dots,A_{i-1}$ and samples obtained $(s^{(1)}_1,s^{(2)}_1),\dots,(s^{(1)}_{i-1}, s^{(2)}_{i-1})$ -- whose labels it does not actually know, \Tester provides:
      \begin{itemize}
        \item a number $\newset$ for each $A\in\operatorname{At}( A_1,\dots,A_{i-1} )$, between $0$ and $\abs{A\setminus \{s^{(1)}_j, s^{(2)}_j\}_{1\leq j \leq i-1}}$ (How many \emph{fresh, not-already-seen} elements of each particular atom $A$ should be included in the next query.)
        \item sets $K_i^{(1)}, K_i^{(2)}\subseteq \{1,\dots, i-1\}$ (Which of the samples $s^{(k)}_1,\dots,s^{(k)}_{i-1}$ will be included in the next query. The labels of these samples are unknown, but are indexed by the index of the query which returned them.)
      \end{itemize} 
    \item based on these specifications, the next query $A_i$ is drawn (but not revealed to \Tester) by
      \begin{itemize}
        \item drawing uniformly at random a set $\Lambda_i$ in 
        \begin{equation}
        \setOfSuchThat{ \Lambda\subseteq [n]\setminus \{s^{(1)}_j, s^{(2)}_j\}_{1\leq j \leq i-1} }{ \forall A\in\operatorname{At}( A_1,\dots,A_{i-1} ),\ \abs{\Lambda\cap A}=\newset }\;.
        \end{equation}
        That is, among all sets, containing only ``fresh elements,'' whose intersection with each atom contains as many elements as \Tester requires.
      \item adding the selected previous samples to this set:
        \begin{equation}
          \Gamma_i \eqdef \setOfSuchThat{ s^{(1)}_j }{  j \in K_i^{(1)} }\cup \setOfSuchThat{ s^{(2)}_j }{  j \in K_i^{(2)} }\ ;\qquad 
          A_i \eqdef \Lambda_i\cup \Gamma_i\ .
        \end{equation}
      \end{itemize}
      This results in a set $A_i$, not fully known to \Tester besides the samples it already got and decided to query again; in which the \emph{labels} of the fresh elements are unknown, but the \emph{proportions} of elements belonging to each atom are known.
    \item samples $s^{(1)}_i\sim(\D_1)_{A_i}$ and $s^{(2)}_i\sim(\D_2)_{A_i}$ are drawn (but not disclosed to \Tester). This defines the $i$-configuration of $A_1,\dots,A_{i}$ and $(s^{(1)}_1,s^{(2)}_1),\dots,(s^{(1)}_{i}, s^{(2)}_{i})$, which is revealed to $\Tester$. Put differently, the algorithm only learns \textsf{(i)} to which of
the $A_\ell$   
the new sample belongs, and \textsf{(ii)} if it is one of the previous samples, in which stage(s) and for which of $\D_1,\D_2$ it has already seen it.
  \end{itemize}
After $t=t(\eps,n)$ such stages, \Tester returns either \accept or \reject, based only on the configuration of $A_1,\dots,A_{t}$ and $(s^{(1)}_1,s^{(2)}_1),\dots,(s^{(1)}_{t}, s^{(2)}_{t})$ (which is all the information it ever had access to).
\end{definition}
Note that in particular, \Tester does not know the labels of samples it got, nor the actual queries it makes: it knows all about their sizes and sizes of their intersections, but not the actual ``identity'' of the elements they contain. 

\subsection{On the use of Yao's Principle in our lower bounds}\label{ssec:prelim:yao}

We recall Yao's Principle (\eg, see Chapter 2.2 of~\cite{MR:95}), a technique which is ubiquitous in the analysis of randomized algorithms.
Consider a set $S$ of instances of some problem: what this principle states is that the worst-case expected cost of a randomized algorithm on instances in $S$ is lower-bounded by the expected cost of the best deterministic algorithm on an instance drawn randomly from $S$.

As an example, we apply it in a standard way in \expref{Section}{sec:lb-uniform}: instead of considering a randomized algorithm working on a fixed instance, we instead analyze a \emph{deterministic} algorithm working on a \emph{random} instance. (We note that, importantly, the randomness in the samples returned by the $\COND$ oracle is ``external'' to this argument, and these samples behave identically in an application of Yao's Principle.)

On the other hand, our application in \expref{Section}{sec:lb-equiv} is slightly different, due to our use of adaptive core testers.
Once again, we focus on deterministic algorithms working on random instances, and the randomness in the samples is external and therefore unaffected by Yao's Principle.
However, we stress that the randomness in the choice of the set $\Lambda_i$ is also external to the argument, and therefore unaffected -- similar to the randomness in the samples, the algorithm has no control here.
Another way of thinking about this randomness is via another step in the distribution over instances: after an instance (which is a pair of distributions) is randomly chosen, we permute the labels on the elements of the distribution's domain uniformly at random.
We note that since the property in question is label-invariant, this does not affect its value.
We can then use the model as stated in \expref{Section}{ssec:prelim:coretesters} for ease of analysis, observing that this can be considered an application of the principle of deferred decisions (as in Chapter 3.5 of \cite{MR:95}).

\subsection{Chernoff bounds for Binomials and Hypergeometrics}\label{ssec:chernoff}
We will make extensive use of Chernoff-style bounds in this work.
Recall that the $\mathrm{Binomial}(n,p)$ distribution describes the distribution of the number of successes when we run $n$ independent Bernoulli trials, each with success probability $p$.

\begin{lemma}[Chernoff Bound for Binomials]\label{lem:chernoff-bin}
Let $X \sim \mathrm{Binomial}(n,p)$ and $\mu = \expect{X} = np$.
Then
\[
  \forall\delta\in(0,1),\qquad \probaOf{|X - \mu| \geq \delta \mu} \leq 2\exp\left(-\frac{\delta^2\mu}{3}\right).
\]
\end{lemma}

We will also need a similar Chernoff-style bound for the \emph{hypergeometric} distribution.
The $\mathrm{Hypergeometric}(n, K, N)$ distribution describes the distribution of the number of successes when we draw $n$ times without replacement from a population of size $N$, in which $K$ objects have the pertinent feature (and thus count as successes).
Note that if the drawing were done with replacement, and $K/N = p$, then this would be equivalent to $\mathrm{Binomial}(n,p)$.
Sampling without replacement introduces \emph{negative correlation} between the probability of each draw being successful.
This type of negative correlation generally ``helps'' with concentration, allowing one to prove similar concentration bounds (see, \eg, \cite{Chvatal79, DubhashiR96}, Theorem 1.17 of~\cite{HyperGeom:Concentration}).

\begin{lemma}[Chernoff Bound for Hypergeometrics]\label{lem:chernoff-hyp}
Let $X \sim \mathrm{Hypergeometric}(n, K, N)$ and $\mu = \expect{X} = nK/N$.
Then,
\[
  \forall\delta\in(0,1),\qquad \probaOf{|X - \mu| \geq \delta \mu} \leq 2\exp\left(-\frac{\delta^2\mu}{3}\right).
\]
\end{lemma}

\section{A Lower Bound for Equivalence Testing}\label{sec:lb-equiv}
We prove our main lower bound on the sample complexity of testing
equivalence between unknown distributions. 
We construct two priors \dyes and \dno over \emph{pairs} of
distributions $(\D_1,\D_2)$ over \domain. \dyes is a distribution over
pairs of distributions of the form $(\D,\D)$, namely the case when the
distributions are identical. Similarly, \dno is a distribution over
$(\D_1,\D_2)$ with $\totalvardist{\D_1}{\D_2}\ge\frac14$. 
We then show that no algorithm \Tester making
$\bigO{\sqrt{\log\log n}}$ queries to
$\COND^{\D_1},\COND^{\D_2}$ can distinguish between a draw from \dyes
and \dno with constant probability (over the choice of $(\D_1,\D_2)$,
the randomness in the samples it obtains, and its internal
randomness).

We describe the construction of $\dyes$ and $\dno$
in~\expref{Section}{ssec:main:lb:construction}, and provide a detailed
analysis in~\expref{Section}{ssec:main:lb:analysis}.  

\subsection{Construction}\label{ssec:main:lb:construction}

We now summarize how a pair of distribution is constructed under $\dyes$ and $\dno$. (Each specific step will be described in more detail in the subsequent paragraphs.)\smallskip
\begin{enumerate}
\item{\bf Effective Support} 
\begin{enumerate}
\item
Pick $k_b$ from the set
  $\{0,1,\ldots,\half\log n\}$ at random. 
\item
Let $b=2^{k_b}$ and $m\eqdef b\cdot n^{1/4}$.
\end{enumerate}
\item{\bf Buckets}
\begin{enumerate}
\item
Choose $\rho $ and $r$ such that $\sum_{i=1}^{2r}\rho^i=n^{1/4}$.
\item
Divide $\{1,\ldots, m\}$ into intervals $B_1, \ldots, B_{2r}$ with $|B_i|=b\cdot
\rho^{i}$. 
\end{enumerate}
\item {\bf Distributions}
\begin{enumerate}
\item
For each $i \in [2r]$, assign probability mass $\frac1{2r}$ uniformly over $B_i$ to generate distribution
$\D_1$. 
\item For each $i \in [r]$ independently, pick $\pi_i$ to be a Bernoulli
trial  
with $\Pr(\pi_i=0)=\half$; if $\pi_i = 0$
then assign probability mass $\frac1{4r}$ and $\frac3{4r}$ over
$B_{2i-1}$ and $B_{2i}$, 
respectively, else $\frac3{4r}$ and $\frac1{4r}$, respectively.
This generates a distribution $\D_2$.
\end{enumerate}
\item {\bf Support relabeling}
\begin{enumerate}
\item
Pick a permutation $\sigma\in S_n$ of the \emph{total} support $n$.
\item
Relabel the symbols of $D_1$ and $D_2$ according to $\sigma$. 
\end{enumerate}
\item{\bf Output:} Generate $(\D_1,\D_1)$ for \dyes, and $(\D_1,\D_2)$
  otherwise. 
\end{enumerate}

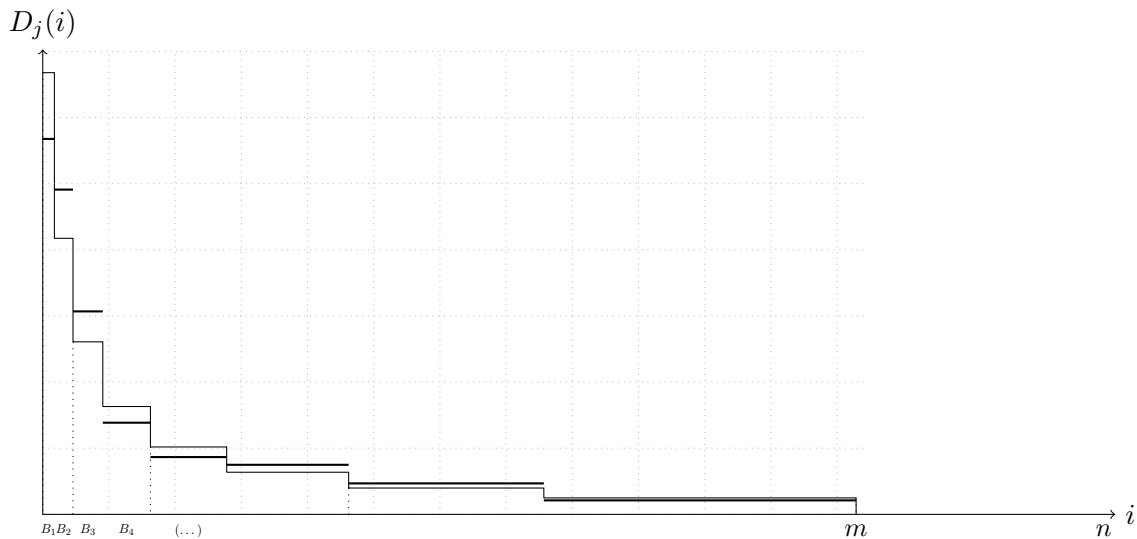
\begin{figure}[!ht]\centering
  \begin{tikzpicture}[x=5pt, y=2pt,scale=0.88]
  \pgfmathsetmacro{\KFACT}{1.6}
  \pgfmathsetmacro{\rbuckets}{8}   \pgfmathsetmacro{\rbucketsminus}{{4*floor((\rbuckets-1)/4)}}
  \pgfmathsetmacro{\rbucketshalf}{{floor(\rbuckets/2)-1}}
  \pgfmathsetmacro{\xmax}{ {(pow(\KFACT,\rbuckets)-1)/(\KFACT-1) + 1} }
  \pgfmathsetmacro{\ymax}{100}
  
  \pgfmathsetmacro{\DFACT}{0.85}
   
         \draw[dotted] \foreach \i in {0,2,...,\rbuckets}{
       ({(pow(\KFACT,\i)-1)/(\KFACT-1)},{(\ymax-5)*pow(1/\KFACT, \i)}) -- ({(pow(\KFACT,\i)-1)/(\KFACT-1)},0)
     };
         \foreach \i in {1,...,4}{
      \node[below] at ({((pow(\KFACT,\i-1)-1)+(pow(\KFACT,\i)-1))/(2*(\KFACT-1))},0) {\scalebox{0.45}{$B_{\i}$}};
    };
    \foreach \i in {5}{
      \node[below] at ({((pow(\KFACT,\i-1)-1)+(pow(\KFACT,\i)-1))/(2*(\KFACT-1))},0) {\scalebox{0.45}{$(\dots)$}};
    };
        \foreach \i in {0,4,...,\rbucketsminus}{
    };
    
    \draw[-] \foreach \i in {1,...,\rbuckets}{
      ({(pow(\KFACT,\i-1)-1)/(\KFACT-1)},{(\ymax-5)*pow(1/\KFACT, \i-1)})
        -- ({(pow(\KFACT,\i)-1)/(\KFACT-1)},{(\ymax-5)*pow(1/\KFACT, \i-1)})
        -- ({(pow(\KFACT,\i)-1)/(\KFACT-1)},{(\ymax-5)*pow(1/\KFACT, \i)})
    };
     \draw[thin] ({(pow(\KFACT,\rbuckets)-1)/(\KFACT-1)},{(\ymax-5)*pow(1/\KFACT, \rbuckets)}) -- ({(pow(\KFACT,\rbuckets)-1)/(\KFACT-1)},0);
    
        \pgfmathsetseed{100}     \foreach \i in {0,...,\rbucketshalf}{
      \pgfmathsetmacro{\cointoss}{round(random)};
      \pgfmathsetmacro{\fstHigh}{{\cointoss*\DFACT + (1-\cointoss)/\DFACT}};
      \pgfmathsetmacro{\sndHigh}{{\cointoss/\DFACT + (1-\cointoss)*\DFACT}};
      
      \pgfmathtruncatemacro{\iIsEven}{int(mod(\i, 2))}
      \ifnum\iIsEven=0
        \xdef\linestyle{thick};
      \else
        \xdef\linestyle{thick};
      \fi
      
      \draw[thin,\linestyle]  ({(pow(\KFACT,2*\i)-1)/(\KFACT-1)},{\fstHigh*(\ymax-5)*pow(1/\KFACT, 2*\i)})
        -- ({(pow(\KFACT,2*\i+1)-1)/(\KFACT-1)},{(\fstHigh*(\ymax-5)*pow(1/\KFACT, 2*\i)});
      \draw[thin,\linestyle]  ({(pow(\KFACT,2*\i+1)-1)/(\KFACT-1)},{\sndHigh*(\ymax-5)*pow(1/\KFACT, 2*\i+1)})
        -- ({(pow(\KFACT,2*(\i+1))-1)/(\KFACT-1)},{(\sndHigh*(\ymax-5)*pow(1/\KFACT, 2*\i+1)});
    };
    
        \draw[ultra thin, help lines, dotted] (0,0) grid (\xmax,\ymax);
    \draw [<->] (0,\ymax) node[above] {$\D_j(i)$} -- (0,0) -- ({1.3*\xmax},0)  node[right] {$i$};
    
    \node[below] at ({\xmax-1},0) {\small $m$};
    \node[below] at ({1.3*\xmax-1},0) {\small $n$};
    
  \end{tikzpicture}\caption{\label{fig:constr:no:inst}A \no-instance $(\D_1,\D_2)$ (before permutation).}
\end{figure}

\noindent We now describe the various steps of the construction in greater
detail.
\paragraph*{Effective support.}
Both $\D_1$ and $\D_2$, albeit distributions on $\domain$, will have
(common) \emph{sparse} support. 
The support size is taken to be $m \eqdef b\cdot n^{1/4}$. 
Note that, from the
above definition, $m$ is chosen uniformly at random from products of
$n^{1/4}$ with powers of $2$, resulting in values in $[n^{1/4},
n^{3/4}]$.\medskip 

In this step $b$ will act as a random scaling factor. The objective of
this random scaling is to induce uncertainty in the algorithm's
knowledge of the true support size of the distributions, and to
prevent it from leveraging this information to test equivalence. 
In fact one can verify that the class of distributions induced for a
single value of $b$, namely all distributions have the same value of
$m$, then 
one can distinguish the \dyes and \dno cases with only $\bigO{1}$
conditional queries. 
The test would (roughly) go as follows.
Since $|B_i|$ is known, one can choose a random subset $S$ of the domain which (with high probability) has no intersection with $B_i$ for $i \leq 2r - 2$, a $O(1)$ size intersection with $B_{2r-1}$, and a $O(\rho)$ size intersection with $B_{2r}$.
Perform $O(1)$ conditional queries over the set $S$, for both distributions.
Given these queries, we can then identify which elements of $S$ belong to $B_{2r-1}$ or $B_{2r}$ -- namely, those which occur at most once belong to $B_{2r}$, and those which occur at least twice belong to $B_{2r-1}$.
In a $\dyes$ instance, then in both distributions, a $1/2$ fraction of queries will belong to $B_{2r-1}$, whereas in a $\dno$ instance, one distribution will have either a $1/4$ or $3/4$ fraction of queries in $B_{2r-1}$, allowing us to distinguish the two cases.

\paragraph*{Buckets.}
Our construction is inspired by the lower bound of Canonne, Ron, and Servedio~\cite[Theorem
8]{CRS:12} for the more restrictive \PCOND access model. 
We partition the support into $2r$ consecutive intervals (henceforth
referred to as \emph{buckets}) $B_1,\dots, B_{2r}$, where the size of
the $i$-th bucket is $b\rho^i$. 
We note that $r$ and $\rho$ will be chosen such that $\sum_{i=1}^{2r} b\rho^i = bn^{1/4}$,
\ie, the buckets fill the effective support.

\paragraph*{Distributions.}
We output a pair of distributions $(\D_1,\D_2)$. Each distribution
that we construct is uniform within any particular bucket $B_i$. In
particular, the first distribution assigns the same mass $1/2r$
to each bucket. Therefore, points within $B_i$ have the same
probability mass $1/2rb \rho^i$. For the $\dyes$ case,
the second distribution is identical to the first. For the $\dno$
case, we pair buckets in $r$ consecutive \emph{\bpairs} $\Pi_1,\dots,\Pi_r$,
with $\Pi_i= B_{2i-1}\cup B_{2i}$. 
For the second distribution $\D_2$, we consider the same buckets as
$\D_1$, but repartition the mass $1/r$ \emph{within} each
$\Pi_i$. More precisely, in each pair, one of the buckets gets now 
total probability mass $1/4r$ while the other gets
$3/4r$ (so that the probability of every point is either
decreased by a factor $1/2$ or increased by
$3/2$). The choice of which goes up and which goes down is
done uniformly and independently at random for each \bpair 
determined by the random choices of
the $\pi_i$.  

\paragraph*{Random relabeling.}
The final step of the construction randomly relabels the symbols,
namely is a random injective map from $[m]$ to $[n]$. This is done to
ensure that no information about the individual symbol labels can be
used by the algorithm for testing. For example, without this the
algorithm can consider a few symbols from the first bucket and
distinguish the \dyes and \dno cases. 
As mentioned in~\expref{Section}{ssec:prelim:yao}, for ease of analysis,
  the randomness in the choice of the permutation is, in some sense,
  deferred to the randomness in the choice of $\Lambda_i$ during
  the algorithm's execution.

\paragraph*{Summary.}
A \no-instance $(\D_1,\D_2)$ is thus defined by the following
parameters: the support size $m$, the vector
$(\pi_1,\dots,\pi_r)\in\{0,1\}^r$
(which only impacts $\D_2$), and the final permutation $\sigma$ of the
domain. A
\yes-instance $(\D_1,\D_1)$ follows an identical process,
however, $\vec \pi$ has no influence on the final outcome.
See~\expref{Figure}{fig:constr:no:inst} for an illustration of such a
$(\D_1,\D_2)$ when $\sigma$ is the identity permutation and thus the
distribution is supported over the first $m$ natural numbers. 

\paragraph*{Values for $\rho$ and $r$.}
  By setting $r={\log n}/{(8\log \rho)} + \bigO{1}$, we have as
  desired $\sum_{i=1}^{2r}\abs{B_i} = m$ 
  and there is a factor $(1+\littleO{1})n^{1/4}$ between the height of
  the first bucket $B_1$ and the one of the last, $B_{2r}$. It remains
  to choose the parameter $\rho$ itself; we shall take it to be
  $2^{\sqrt{\log n}}$, resulting in $r=\frac{1}{8}\sqrt{\log n}
    + \bigO{1}$. (Note that for the sake of the exposition, we ignore
  technical details such as the rounding of parameters, \eg, bucket
  sizes; these can be easily taken care of at the price of cumbersome
  case analyses, and do not bring much to the argument.)

\subsection{Analysis}\label{ssec:main:lb:analysis}

We now prove our main lower bound, by analyzing the behavior of core adaptive testers (as per~\expref{Definition}{def:core:tester}) on the families \dyes and \dno from the previous section.
In~\expref{Section}{sssec:ban}, we argue that, with high probability, the sizes of the queries performed by the algorithm satisfy some specific properties.
Conditioned upon this event, in~\expref{Section}{sssec:induction}, we show that the algorithm will get similar information from each query, whether it is running on a $\yes$-instance or a $\no$-instance.

Before moving to the heart of the argument, we state the following fact, straightforward from the construction of our \no-instances.
\begin{fact}\label{fact:main:lb:dist:d1:d2}
  For any $(\D_1,\D_2)$ drawn from $\dno$, one has $\totalvardist{\D_1}{\D_2} = 1/4$.
\end{fact}
\noindent Moreover, as allowing more queries can only increase the probability of success, we hereafter focus on a core adaptive tester that performs exactly $q=\frac{1}{10}\sqrt{\log\log n}$ (adaptive) queries; and will show that it can only distinguish between \yes and \no-instances with probability $o(1)$.

\subsubsection{Banning ``bad queries''}\label{sssec:ban}
As mentioned in~\expref{Section}{ssec:main:lb:construction}, the draw of a \yes or \no-instance involves a random scaling of the size of the support of the distributions, meant to ``blind'' the testing algorithm.
Recall that a testing algorithm is specified by a decision tree, which at step $i$, specifies how many unseen elements from each atom to include in the query ($\{\newset\}$) and which previously seen elements to include in the query (sets $K^{(1)}_i, K^{(2)}_i$, as defined in~\expref{Section}{ssec:prelim:coretesters}), where the algorithm's choice depends on the observed configuration at that time.
Note that, using Yao's Principle (as discussed in~\expref{Section}{ssec:prelim:yao}), these choices are deterministic for a given configuration -- in particular, we can think of all $\{\newset\}$ and $K^{(1)}_i, K^{(2)}_i$ in the decision tree as being fixed.
In this section, we show that all $\newset$ values satisfy with high probability some particular conditions with respect to the choice of distribution, where the randomness is over the choice of the support size.\medskip

First, we recall an observation from \cite{CFGM:13}, though we modify it slightly to apply to configurations on pairs of distributions and we apply a slightly tighter analysis.
This essentially limits the number of states an algorithm could be in by a function of how many queries it makes.
\begin{proposition}\label{prop:tree-nodes}
  The number of nodes in a decision tree corresponding to a $q$-sample algorithm is at most $2^{6q^2 + 1}$.
\end{proposition}
\begin{proof}
  As mentioned in~\expref{Definition}{def:configuration}, an $i$-configuration can be described using $6i^2$ bits, resulting in at most $2^{6i^2}$ $i$-configurations.
  Since each $i$-configuration leads us to some node on the $i$-th level of the decision tree, the total number of nodes can be upper bounded by summing over the number of $i$-configurations for $i$ ranging from $0$ to $q$, giving us the desired bound.
\end{proof}

For the sake of the argument, we will introduce a few notions applying to the \emph{sizes} of query sets: namely, the notions of a number being \emph{small}, \emph{large}, or \emph{stable}, and of a vector being \emph{incomparable}.
Roughly speaking, a number is small if a uniformly random set of this size does not, in expectation, hit the largest bucket $B_{2r}$ -- in other words, the set is likely to be disjoint from the support.
On the other hand, it is large if we expect such a set to intersect many \bpairs (\ie, a significant fraction of the support). 

The definition of stable numbers is slightly more quantitative: a number $\beta$ is stable if a random set of size $\beta$, for each bucket $B_i$, is either disjoint from $B_i$ or has an intersection with $B_i$ of size close to the expected value.
In the latter case, we say the set \emph{concentrates} over $B_i$. 
Finally, a vector of values $(\beta_j)$ is incomparable if the union of random sets $S_1,\dots,S_m$ of sizes $\beta_1,\dots,\beta_m$ contains (with high probability) an amount of mass $\D\left(\bigcup_j S_j\right)$ which is either much smaller or much larger than the probability $\D(s)$ of any single element $s$.

\noindent We formalize these concepts in the definitions below. To motivate them, it will be useful to bear in mind that, from the construction described in~\expref{Section}{ssec:main:lb:construction}, the expected intersection of a uniform random set of size $\beta$ with a bucket $B_i$ is of size $\beta b\rho^i/n$; while the expected probability mass from $B_i$ it contains (under either $\D_1$ or $\D_2$) is $\beta/2rn$.
\begin{definition}
  Let $q$ be an integer, and let $\varphi=\Theta(q^{5/2})$.  A number $\beta$ is said to be \emph{small} if $\beta < \frac{n}{b \rho^{2r}}$; it is \emph{large} (with relation to some integer $q$) if $\beta \geq \frac{n}{b \rho^{2r-2\varphi}}$.
\end{definition}
Note that the latter condition equivalently means that, in expectation, a set of large size will intersect at least $\varphi+1$ \bpairs (as it hits an expected $2\varphi+1$ buckets, since $\beta \abs{ B_{2r-2\varphi} }/n \geq 1$).
From the above definitions we get that, with high probability, a random set of any fixed size will in expectation either hit many or no buckets.
\begin{proposition}\label{prop:small-large}
  A number is either small or large with probability $1 - \bigO{ \frac{\varphi \log \rho}{\log n} }$.
\end{proposition}
\begin{proof}
  A number $\beta$ is neither large nor small if $\frac{\rho^{2\varphi}n}{\beta\rho^{2r}} \leq b \leq \frac{n}{\beta\rho^{2r}}$.
  The ratio of the endpoints of the interval is $\rho^{2\varphi}$.
  Since $b = 2^{k_b}$, this implies that at most $\log \rho^{2\varphi} = 2\varphi \log \rho$ values of $k_b$ could result in a fixed number falling in this range.
  As there are $\bigTheta{\log n}$ values for $k_b$, the proposition follows.
\end{proof}

The next definition characterizes the sizes of query sets for which the expected intersection with any bucket is either close to 0 (less than $1/\alpha$, for some threshold $\alpha$), or very big (more than $\alpha$). (It will be helpful to keep in mind that we will eventually use this definition with $\alpha=\poly(q)$.)
\begin{definition}
  A number $\beta$ is said to be \emph{$\alpha$-stable} (for $\alpha \geq 1$) if, for each $j \in [2r]$, $\beta \notin\left[ \frac{n}{\alpha b \rho^j}, \frac{\alpha n}{b \rho^j} \right]$. 
  A vector of numbers is said to be $\alpha$-stable if all numbers it contains are $\alpha$-stable.
\end{definition}
\begin{proposition}\label{prop:stable}
  A number is $\alpha$-stable with probability $1 - \bigO{ \frac{r \log \alpha}{\log n} }$.
\end{proposition}
\begin{proof}
  Fix some $j \in [2r]$.
  A number $\beta$ does not satisfy the definition of $\alpha$-stability for this $j$ if $\frac{n}{\alpha \beta \rho^j} \leq b \leq \frac{n\alpha}{\beta \rho^j}$.
  Since $b = 2^{k_b}$, this implies that at most $\log 2\alpha$ values of $k_b$ could result in a fixed number falling in this range.
  Noting that there are $\Theta(\log n)$ values for $k_b$ and taking a union bound over all $2r$ values for $j$, the proposition follows.
\end{proof}

The following definition characterizes the sizes of query sets which have a probability mass far from the probability mass of any individual element. (For the sake of building intuition, the reader may replace $\nu$ in the following by the parameter $b$ of the distribution.)

\begin{definition}\label{def:incomparable}
  A vector of numbers $(\beta_1, \dots, \beta_\ell)$ is said to be \emph{$(\alpha,\tau)$-incomparable with respect to $\nu$} (for $\tau \geq 1$) if the two following conditions hold.
  \begin{itemize}
    \item $(\beta_1, \dots, \beta_\ell)$ is $\alpha$-stable.
    \item Let $\Delta_j$ be the minimum $\Delta \in \{0,\dots, 2r\}$ such that $\frac{\beta_j \nu \rho^{2r - \Delta}}{n} \leq \frac{1}{\alpha}$, or $2r$ if no such $\Delta$ exists.
      For all $i \in [2r]$, $\frac{1}{2rn} \sum_{j = 1}^\ell \beta_j \Delta_j \not \in \left[\frac{1}{\tau 2r\nu \rho^i}, \frac{\tau}{2r \nu \rho^i}\right]$.
  \end{itemize}
\end{definition}
\noindent Recall from the definition of $\alpha$-stability of a number that a random set of this size either has essentially no intersection with a bucket or ``concentrates over it'' (\ie, with high probability, the probability mass contained in the intersection with this bucket is very close to the expected value).
The above definition roughly captures the following. For any $j$, $\Delta_j$ is the number of buckets that will concentrate over a random set of size $\beta_j$.
The last condition asks that the total probability mass from $\D_1$ (or $\D_2$) enclosed in the union of $m$ random sets of size  $\beta_1,\dots,\beta_\ell$ be a multiplicative factor of $\tau$ from the individual probability weight $1/2r b \rho^i$ of a single element from any of the $2r$ buckets.

\begin{proposition}\label{prop:incomparable}
  Given that a vector of numbers of length $\ell$ is $\alpha$-stable, it is $(\alpha,q^2)$-incomparable with respect to $b$ with probability at least $1 - \bigO{\frac{r\log q}{\log n}}$.
\end{proposition}
\begin{proof}
Fix any vector $(\beta_1,\dots,\beta_\ell)$. By the definition above, for each value $b$ such that $(\beta_1,\dots,\beta_\ell)$ is $\alpha$-stable, we have
\[
     \beta_j\cdot\frac{\alpha\rho^{2r}}{n} \leq \frac{\rho^{\Delta_j}}{b} < \beta_j\cdot\frac{\alpha\rho^{2r+1}}{n}, \quad j\in[\ell]
\]
or, equivalently,
\[
     \frac{\log\frac{\alpha\beta_j}{n}}{\log \rho} + 2r + \frac{\log b}{\log \rho}
     \leq \Delta_j
     < \frac{\log\frac{\alpha\beta_j}{n}}{\log \rho} + 2r + \frac{\log b}{\log \rho}   + 1, 
     \quad j\in[\ell].
\]
Writing $\lambda_j\eqdef \frac{\log\frac{\alpha\beta_j}{n}}{\log \rho} + 2r$ for $j\in[\ell]$, we obtain that 
\begin{equation}\label{eq:incomparable:stability:consequence}
    \sum_{j=1}^\ell \beta_j \Delta_j b = b \sum_{j=1}^\ell \beta_j (\lambda_j+\bigO{1}) + \frac{b\log b}{\log \rho}\sum_{j=1}^\ell \beta_j.
\end{equation}
\begin{itemize}
  \item If it is the case that $\log \rho \cdot \sum_{j=1}^\ell \beta_j (\lambda_j+\bigO{1}) \ll \log b\cdot \sum_{j=1}^\ell \beta_j$. Then, for any fixed $i\in[2r]$, to meet the second item of the definition of incomparability we need
$\sum_{j=1}^\ell \beta_j \Delta_j b \notin [n/(200q \rho^i), 200qn/\rho^i]$. This is essentially, with the assumption~\expeqref{equation}{eq:incomparable:stability:consequence} above, requiring that
\[
    b\log b \notin  \left [\frac{n\log \rho}{2q^2 \rho^i \sum_{j=1}^\ell \beta_j}, \frac{2q^2 n\log \rho}{\rho^i \sum_{j=1}^\ell \beta_j} \right].
\] 
Recalling that $b\log b = k_b 2^{k_b}$, this means that $\bigO{\log q/\log \log q}$ values of $k_b$ are to be ruled out. (Observe that this is the number of possible ``bad values'' for $b$ without the condition from the case distinction above; since we add an extra constraint on $b$, there are at most this many values to avoid.)
  \item Conversely, if $\log \rho \cdot \sum_{j=1}^\ell \beta_j (\lambda_j+\bigO{1}) \gg \log b \cdot \sum_{j=1}^\ell \beta_j$ the requirement becomes
  \[
    b \notin  \left [\frac{n\log \rho}{2q^2 \rho^i \sum_{j=1}^\ell \beta_j(\lambda_j+\bigO{1})}, \frac{2q^2 n\log \rho}{\rho^i \sum_{j=1}^\ell \beta_j(\lambda_j+\bigO{1})} \right].
  \] 
ruling out this time $\bigO{\log q}$ values for $k_b$.
  \item Finally, the two terms are comparable only if $\log b = \Theta\Big( \log \rho \cdot \sum_{j=1}^\ell \beta_j (\lambda_j+\bigO{1}) \cdot \left(\sum_{j=1}^\ell \beta_j\right)^{-1} \Big)$; given that $\log b=k_b$, this rules out this time $\bigO{1}$ values for $k_b$.
\end{itemize}
A union bound over the $2r$ possible values of $i$, and the fact that $k_b$ can take $\bigTheta{\log n}$ values, complete the proof.
\end{proof}

\noindent We put these together to obtain the following lemma.
\begin{lemma}\label{lem:banned}
  With probability at least $1 - \bigO{\frac{2^{6q^2 + q}(r \log \alpha + \varphi \log \rho) +  2^{6q^2}(r \log q)}{\log n} }$, the following holds for the decision tree corresponding to a $q$-query algorithm:
  \begin{itemize}
    \item the size of each atom is $\alpha$-stable and either large or small;
    \item the size of each atom, after excluding elements we have previously observed,\footnote{More precisely, we mean to say that for each $i \leq q$, for every atom $A$ defined by the partition of $(A_1, \dots, A_i)$, the values $k_i^A$ and $|A \setminus \{s_1^{(1)}, s_1^{(2)}, \dots, s_{i-1}^{(1)}, s_{i-1}^{(2)}\}| - k_i^A$ are $\alpha$-stable and either large or small; } is $\alpha$-stable and either large or small;
    \item for each $i$, the vector $(\newset)_{\atom\in\operatorname{At}(A_1,\dots, A_i)}$ is $(\alpha,q^2)$-incomparable (with respect to $b$).
  \end{itemize}
\end{lemma}
\begin{proof}
  From~\expref{Proposition}{prop:tree-nodes}, there are at most $2^{6q^2 + 1}$ tree nodes, each of which contains one vector $(\newset)_A$, and at most $2^{q}$ atom sizes.
  The first point follows from~\expref{Proposition}{prop:small-large} and~\expref{Proposition}{prop:stable} and applying the union bound over all $2^{6q^2 + 1} \cdot 2 \cdot 2^q$ sizes, where we note the additional factor of $2$ comes from either including or excluding the old elements. 
  The latter point follows from~\expref{Proposition}{prop:incomparable} and applying the union bound over all $2^{6q^2 + 1}$ nodes in the tree (each containing a single $\newset$ vector).
\end{proof}
 
\subsubsection{Key lemma: bounding the variation distance between decision trees}\label{sssec:induction}
In this section, we prove a key lemma on the variation distance between the distribution on leaves of any decision tree, when given access  to either an instance from \dyes or \dno. This lemma will in turn directly yield~\expref{Theorem}{theo:main:lb:equiv}. Hereafter, we set the parameters $\alpha$ (the threshold for stability), $\varphi$ (the parameter for smallness and largeness) and $\gamma$ (an accuracy parameter for how well things concentrate over their expected value) as follows.\footnote{This choice of parameters is not completely arbitrary: combined with the setting of $q$, $r$ and $\rho$, they ensure a total bound $o(1)$ on variation distance and probability of ``bad events'' as well as a (relative) simplicity and symmetry in the relevant quantities.} $\alpha\eqdef q^7$, $\varphi\eqdef q^{5/2}$ and $\gamma\eqdef 1/\varphi = q^{-5/2}$. (Recall further that $q=\frac{1}{10}\sqrt{\log\log n}$.)
\begin{lemma}\label{lem:induction}
  Conditioned on the events of~\expref{Lemma}{lem:banned}, consider the distribution over leaves of any decision tree corresponding to a $q$-query adaptive algorithm when the algorithm is given a $\yes$-instance, and when it is given a $\no$-instance.
  These two distributions have total variation distance $\littleO{1}$.
\end{lemma}
\begin{proof}
  This proof is by induction on $1\leq i\leq q$.
  We will have three inductive hypotheses, $\iha{t}, \ihb{t},$ and $\ihc{t}$.
  Assuming all three hold for all $t < i$, we prove $\iha{i}$.
  Additionally assuming $\iha{i}$, we prove $\ihb{i}$ and $\ihc{i}$.

  Roughly, the first inductive hypothesis states that the query sets behave similarly to as if we had picked a random set of that size.
    It also implies that whether or not we get an element we have seen before is ``obvious'' based on past observances and the size of the query we perform.
    The second states that we never observe two distinct elements from the same bucket-pair.
    The third states that the next sample is distributed similarly in either a $\yes$-instance or a $\no$-instance.
    Note that this distribution includes both features which our algorithm can observe (\ie, the atom which the sample belongs to and if it collides with a previously seen sample), as well as those which it can not (\ie, which bucket-pair the observed sample belongs to).
    It is necessary to show the latter, since the bucket-pair a sample belongs to may determine the outcome of future queries.

  More precisely, the three inductive hypotheses are as follows:
  \begin{itemize}
    \item $\iha{i}$:
      In either a $\yes$-instance or a $\no$-instance, the following occurs:
      For an atom $S$ in the partition generated by $A_1, \dots, A_i$, let $S' = S\ \setminus\ \{s_1^{(1)}, s_1^{(2)}, \dots, s_{i-1}^{(1)}, s_{i-1}^{(2)}\}$.
      For every such $S'$, let $\ell^{S'}$ be the largest index $\ell \in \{0, \dots, 2r\}$ such that $\frac{|S'|b\rho^\ell}{n} \leq \frac{1}{\alpha}$, or $0$ if no such $\ell$ exists.
      We claim that $\ell^{S'} \in \{0, \dots, 2r - \varphi - 2\} \cup \{2r\}$, and say $S'$ is small if $\ell^{S'} = 2r$ and large otherwise.
      Additionally:
      \begin{itemize}
        \item for $j \leq \ell^{S'}$, $\abs{S' \cap B_j } = 0$;
        \item for $j > \ell^{S'}$, $\abs{S' \cap B_j}$ lies in  $\left[1 - i\gamma, 1 + i\gamma\right] \frac{\abs{S'}b\rho^j}{n}$.
      \end{itemize}
      Furthermore, let $p_1$ and $p_2$ be the probability mass contained in $\Lambda_i$ and $\Gamma_i$, respectively.
      Then $\frac{p_1}{p_1 + p_2} \leq \bigO{\frac{1}{q^2}}$ or $\frac{p_2}{p_1 + p_2} \leq \bigO{\frac{1}{q^2}}$ (that is, either almost all the probability mass comes from elements which we have not yet observed, or almost all of it comes from previously seen ones).
 
    \item $\ihb{i}$: 
      No two elements from the set $\{s_1^{(1)}, s_1^{(2)}, \dots, s_i^{(1)}, s_i^{(2)}\}$ belong to the same bucket-pair.
    \item $\ihc{i}$: 
      Let $T^{\yes}_i$ be the random variable representing the atoms and bucket-pairs\footnote{If a sample $s_i^{(k)}$ does not belong to any bucket (if the corresponding $i$-th query did not intersect the support), it is marked in $T^{\yes}_i$ with a ``dummy label'' to indicate so.} containing $(s_i^{(1)}, s_i^{(2)})$, as well as which of the previous samples they intersect with, when the $i$-th query is performed on a $\yes$-instance, and define $T^{\no}_i$ similarly for $\no$-instances.
      Then $\totalvardist{T^{\yes}_i }{ T^{\vphantom{yes}\no}_i } \leq \bigO{1/q^2 + {1}/{\rho} + \gamma + {1}/{\varphi}} =\littleO{1}$.
  \end{itemize}
  We will show that $\iha{i}$ holds with probability $1 - \bigO{2^i\exp(-{2\gamma^2\alpha}/{3})}$ and $\ihb{i}$ holds with probability $1 - \bigO{{i}/{\varphi}}$.
  Let $T^{\yes}$ be the random variable representing the $q$-configuration and the bucket-pairs containing each of the observed samples in a $\yes$-instance, and define $T^{\no}$ similarly for a $\no$-instance.
  We note that this random variable determines which leaf of the decision tree we reach, which $\ihc{q}$ bounds.
  We can then take a union bound over all $i \in [q]$ to upper bound the probability that $\iha{i}$ and $\ihb{i}$ do not hold, and use $\ihc{i}$ and the coupling interpretation of total variation distance to upper bound the probability that $T^{\yes}_i$ and $T^{\no}_i$ ever differ.
  Any of these ``failure'' events happens with probability at most         $O({2^q\exp(-\frac{2\gamma^2\alpha}{3}) + \frac{q^2}{\varphi} + \frac{1}{q} + \frac{q}{\rho} + q\gamma + \frac{q}{\varphi}}) = \littleO{1}$ (from our choice of $\alpha,\gamma,\varphi$).
  This upper bounds the total variation distance between $T^{\yes}$ and $T^{\no}$, giving the desired result.

   We proceed with the inductive proofs of $\iha{i}$, $\ihb{i}$, and $\ihc{i}$, noting that the base cases hold trivially for all three of these statements.
   Throughout this proof, recall that $\Lambda_i$ is the set of unseen support elements which we query, and $\Gamma_i$ is the set of previously seen support elements which we query.
  \begin{lemma}\label{lem:induction:2}
    Assume that $\iha{t}, \ihb{t}, \ihc{t}$ hold for all $1 \leq t \leq i-1$.
    Then $\iha{i}$ holds with probability at least 
    $1 - \bigO{2^i\exp\left(-\frac{2\gamma^2\alpha}{3}\right)} = 1 - 2^{i-\bigOmega{q^2}}$.
  \end{lemma}
  \begin{proof}
    We start with the first part of $\iha{i}$ (the statement prior to ``Furthermore'').
    Let $S$ (and the corresponding $S^\prime$) be any atom as in $\iha{i}$.
    First, we note that $\ell^{S'} \in \{0, \dots, 2r - \varphi - 2\} \cup \{2r\}$ since we are conditioning on~\expref{Lemma}{lem:banned}: $|S'|$ is $\alpha$-stable and either large or small, which enforces this condition. 

\noindent Next, suppose $S'$ is contained in some other atom $T$ generated by $A_1,\dots, A_{i-1}$, and let $T' = T \setminus \{s_1^{(1)}, s_1^{(2)}, \dots, s_{i-1}^{(1)}, s_{i-1}^{(2)}\}$.
    Since $|S'| \leq |T'|$, this implies that $\ell^{T'} \leq \ell^{S'}$.
    We argue about $|T' \cap B_{j}|$ for three regimes of $j$:
    \begin{itemize}
      \item The first case is $j \leq \ell^{T'}$.
      By the inductive hypothesis, $\abs{T' \cap B_{j}} = 0$, so $\abs{S' \cap B_{j}} = 0$ with probability $1$.
      \item The next case is $\ell^{T'} < j \leq \ell^{S'}$.
        Recall from the definition of a core adaptive tester that $S'$ will be chosen uniformly at random from all subsets of $T'$ of appropriate size.
        By the inductive hypothesis, $$\frac{\abs{T' \cap B_j}}{\abs{T'}} \in \left[1 - (i-1)\gamma, 1 + (i-1)\gamma\right] \frac{b\rho^j}{n},$$
        and therefore 
        \[
        \expect{\abs{S' \cap B_j}} \in \left[1 - (i-1)\gamma, 1 + (i-1)\gamma\right]\frac{\abs{S'}b\rho^j}{n},\; \text{ implying } \expect{\abs{S' \cap B_j}} \leq \frac{2}{\alpha\rho^{\ell^{S'} - j}};
        \]
        where the inequality is by the definition of $\ell^{S'}$ and using the fact that $(i-1)\gamma \leq 1$.
        Using a Chernoff bound for hypergeometric random variables (Lemma~\ref{lem:chernoff-hyp}), and writing $\mu\eqdef \expect{\abs{S' \cap B_j}} $ for conciseness,
        \begin{align*}
          \probaOf{ \abs{S' \cap B_j} \geq 1 } &= \probaOf{  \abs{S' \cap B_j} \geq \left(1 + \frac{1-\mu}{\mu}\right)\mu } \\
                                                   &\leq \exp\!\left(-\frac{(1-\mu)^2}{3\mu}\right) \\
                                                   &\leq \exp\!\left(-\frac{1}{12}\alpha\rho^{\ell^{S'}-j}\right),
        \end{align*}
        where the second inequality holds because $\mu \leq 2/\alpha\rho^{\ell^{S'}-j}$ and $(1 - \mu)^2 \geq 1/2$ for $n$ sufficiently large.
      \item The final case is $j > \ell^{S'}$.
        As in the previous one, 
        \[
        \expect{\abs{S' \cap B_j}} \in \left[1 - (i-1)\gamma, 1 + (i-1)\gamma\right] \frac{\abs{S'}b\rho^j}{n},\; \text{ implying } \expect{\abs{S' \cap B_j}} \geq \frac{\alpha \rho^{j - \ell^{S'} - 1}}{2};
        \]
        where the inequality is by the definition of $\ell^{S'}$, $\alpha$-stability, and using the fact that $(i-1)\gamma \leq 1/2$.
        Using again a Chernoff bound for hypergeometric random variables (Lemma~\ref{lem:chernoff-hyp}),
        \begin{align*}
          \probaOf{ \abs{S' \cap B_j} - \expect{\abs{S' \cap B_j}} \geq \gamma \frac{\abs{S'}b\rho^j}{n} } 
          &\leq \probaOf{  \abs{S' \cap B_j} - \expect{\abs{S' \cap B_j}} \geq \gamma 2\expect{\abs{S' \cap B_j}} } \\
          &\leq 2\exp\!\left(-\frac{(2\gamma)^2 \expect{\abs{S' \cap B_j}}}{3}\right) \\
          &\leq 2\exp\!\left(-\frac{2}{3}\gamma^2\alpha\rho^{j - \ell^{S'} - 1}\right),
        \end{align*}
        where the first inequality comes from $2(1 - (i-1)\gamma) \geq 1$, the second is from Chernoff bound, and the third follows from $\expect{\abs{S' \cap B_j}} \geq \alpha\rho^{j - \ell^{S'} - 1}/2$.
    \end{itemize}
    Since we wish to prove the statement for all buckets $B_j$ simultaneously, we take a union bound over all $j$. Recall that we want to bound the probability that $S'$ satisfies two conditions, for $j \leq \ell^{S'}$, $\abs{S' \cap B_j } = 0$; and  for $j > \ell^{S'}$, $\abs{S' \cap B_j}$ lies in  $\left[1 - i\gamma, 1 + i\gamma\right] \frac{\abs{S'}b\rho^j}{n}$.
    Using bounds from all three of these regimes of $j$, and a union bound, the probability that $S'$ does not satisfy the conditions of $\iha{i}$ is at most
    \[
    \sum_{j \leq \ell^{T'}} 0 + \sum_{\ell^{T'} < j \leq \ell^{S'}} \exp\!\left(-\frac{1}{12}\alpha\rho^{\ell^{S'}-j}\right) + \sum_{ j > \ell^{S'}}2\exp\!\left(-\frac{2}{3}\gamma^2\alpha\rho^{j - \ell^{S'} - 1}\right).
    \]
    This probability is maximized when $\ell^{S'} = \ell^{T'} = 0$, in which case it is
    $$\sum_{ j = 1}^{2r}2\exp\!\left(-\frac{2}{3}\gamma^2\alpha\rho^{j -  1}\right) \leq \sum_{ j = 1}^{\infty}2\exp\!\left(-\frac{2}{3}\gamma^2\alpha\rho^{j -  1}\right) \leq 3\exp\!\left(-\frac{2}{3}\gamma^2\alpha\right).$$
    Taking a union bound over at most $2^i$ sets gives us the desired probability bound.\medskip

    Finally, we prove the remainder of $\iha{i}$ (the statement following ``Furthermore''); this will follow from the definition of incomparability (Definition~\ref{def:incomparable}).
    \begin{itemize}
      \item First, we focus on $\Gamma_i$. Suppose that $\Gamma_i$ contains at least one element with positive probability mass (if not, the statement trivially holds).
    Let $p_2'$ be the probability mass of the heaviest element in $\Gamma_i$.
    Since our inductive hypothesis implies that $\Gamma_i$ has no elements in the same bucket pair, the maximum possible value for $p_2$ is
    \begin{align*}
      p_2 &\leq p_2' + \frac{3p_2'}{\rho} + \frac{3p_2'}{\rho^3} + \dots \leq p_2' + \frac{3p_2'}{\rho}\sum_{k=0}^\infty \frac{1}{\rho^{2k}} 
          = \left(1+ \frac{3}{\rho}\frac{\rho^2}{\rho^2 -1}\right)p_2' \\
          &\leq (1 + o(1))p_2'
    \end{align*}
    Therefore, $p_2 \in [p_2', (1 + o(1))p_2']$.
    Supposing this heaviest element belongs to bucket $j$, we can say that $p_2 \in [\frac{1}{2}, (1 + o(1))\frac{3}{2}]\frac{1}{2rb\rho^j}$.
    \item Next, we focus on $\Lambda_i$.
    Consider some atom $A$, from which we selected $k_A$ elements which have not been previously observed: call the set of these elements $A'$.
    In the first part of this proof, we showed that for each bucket $B_k$, either $\abs{ A' \cap B_k } = 0$ or $\abs{A' \cap B_k} \in \left[1 - i\gamma, 1 + i\gamma\right] {\abs{A'}b\rho^k}/{n}$.
    In the latter case, noting that $i\gamma \leq \frac12$ and that the probability of an individual element in $B_k$ is within $[1,3]\frac{1}{4rb\rho^k}$, the probability mass contained by $|A' \cap B_k|$ belongs to $[1,9]\frac{\abs{A'}}{8rn}$.
    Recalling the definition of $\Delta_A$ as stated in~\expref{Definition}{def:incomparable}, as shown earlier in this proof, this non-empty intersection happens for exactly $\Delta_A$ buckets.
    Therefore, the total probability mass in $\Lambda_i$ is in the interval $\left[\frac{1}{4},\frac{9}{4}\right]\frac{1}{2rn}\sum_{\atom\in\operatorname{At}(A_1,\dots, A_i)} \newset \Delta_A$.
    \end{itemize}
    Recall that we are conditioning on~\expref{Lemma}{lem:banned} which states that the vector $(\newset)_{\atom\in\operatorname{At}(A_1,\dots, A_i)}$ is $(\alpha, q^2)$-incomparable with respect to $b$.
    Applying this definition to the bounds just obtained on the probability masses in $\Lambda_i$ and $\Gamma_i$ gives~\expref{Lemma}{lem:induction:2}.
 \end{proof}

  \begin{lemma} 
    Assume that $\iha{t}, \ihb{t}, \ihc{t}$ hold for all $1 \leq t \leq i-1$, and additionally $\iha{i}$ holds.
    Then $\ihb{i}$ holds with probability at least $1 - \bigO{\frac{i}{\varphi}}$.
  \end{lemma}
 
 \begin{proof}
   We focus on $s_i^{(1)}$.
   If $s_i^{(1)} \in \Gamma_i$, the conclusion is trivial, so suppose $s_i^{(1)} \in \Lambda_i$.
   From $\iha{i}$, no small atom intersects any of the buckets, so let us condition on the fact that $s_i^{(1)}$ belongs to some large atom $S$.
   Since we want $s_i^{(1)}$ to fall in a distinct bucket-pair from $2(i-1) + 1$ other samples, there are at most $2i-1$ bucket-pairs which $s_i^{(1)}$ should not land in.
   Using $\iha{i}$, the maximum probability mass contained in the intersection of these bucket-pairs and $S$ is $(1+ i\gamma)(2i-1)|S|/rn$.
   Similarly, using the definition of a large atom, the minimum probability mass contained in $S$ is $(1 - i\gamma)\varphi|S|/rn$.
   Taking the ratio of these two terms gives an upper bound on the probability of breaking this invariant, conditioned on landing in $S$, as $\bigO{i/\varphi}$, where we note that $\frac{1 + i\gamma}{1 - i\gamma} = \bigO{1}$.
   Since the choice of which large atom was arbitrary, we can remove the conditioning.
   Taking a union bound for $s_i^{(1)}$ and $s_i^{(2)}$ gives the result.
 \end{proof}

 \begin{lemma}\label{lemma:inductioni} 
   Assume that $\iha{t}, \ihb{t}, \ihc{t}$ hold for all $1 \leq t \leq i-1$, and additionally $\iha{i}$ holds.
   Then $\ihc{i}$ holds.
 \end{lemma}
 \begin{proof}
   We fix some setting of the intersection history, \ie, the configuration and the bucket-pairs the past elements belong to, and show that the results of the next query will behave similarly, whether the instance is a $\yes$-instance or a $\no$-instance.
   We note that, since we are assuming the inductive hypotheses hold, certain settings which violate these hypotheses are not allowed.
   We also note that $s_i^{(1)}$ is distributed identically in both instances, so we focus on $s_i \eqdef s_i^{(2)}$ for the remainder of this proof.

   \noindent First, we condition that, based on the setting of the past history, $s_i$ will either come from $\Lambda_i$ or $\Gamma_i$ -- this event happens with probability $1 - \bigO{1/q^2}$.
   \begin{proposition}
     In either a $\yes$-instance or a $\no$-instance, $s_i$ will either come from $\Lambda_i$ with probability $1 - \bigO{\frac{1}{q^2}}$, or $\Gamma_i$ with probability $1 - \bigO{\frac{1}{q^2}}$, where the choice of which one is deterministic based on the fixed configuration and choice for the bucket-pairs of previously seen elements.
   \end{proposition}
   \begin{proof}
     This is simply a rephrasing of the portion of $\iha{i}$ following ``Furthermore.''
   \end{proof}
   We now try to bound the total variation distance between $T^{\yes}_i$ and $T^{\no}_i$ conditioning on this event.
   In the case when it does not hold, we trivially bound the total variation distance by $1$, incurring a cost of $\bigO{1/q^2}$ to the total variation distance between the unconditioned variables.
   Since our target was for this quantity was $\bigO{1/q^2 + {1}/{\rho} + \gamma + {1}/{\varphi}}$, it remains to show, in the conditioned space, that the total variation distance in either case is at most $\bigO{{1}/{\rho} + \gamma + {1}/{\varphi}} = O(1/q^{5/2})$. 
   We break this into two cases, the first being when $s$ comes from $\Gamma_i$.
   In this case, we incur a cost in total variation distance which is $\bigO{{1}/{\rho}}$:
   \begin{proposition}
     In either a $\yes$-instance or a $\no$-instance, condition that $s_i$ comes from $\Gamma_i$. Then one of the following holds:
     \begin{itemize}
       \item $|\Gamma_i \cap B_j| = 0$ for all $j \in [2r]$, in which case $s_i$ is distributed uniformly at random from the elements of $\Gamma_i$;
       \item or $|\Gamma_i \cap B_j| \neq 0$ for some $j \in [2r]$, in which case $s_i$ will be equal to some $s \in \Gamma_i$ with probability $1 - \bigO{1/\rho}$,
         where the choice of $s$ is deterministic based on the fixed configuration and choice for the bucket-pairs of previously seen elements.
     \end{itemize}
   \end{proposition}
   \begin{proof}
     The former case follows from the definition of the sampling model.      For the latter case, let $p$ be the probability mass of the heaviest element in $\Gamma_i$.
     Since our inductive hypothesis implies that $\Gamma_i$ has no elements in the same bucket-pair, the maximum possible value for the rest of the elements is 
     \[
     \frac{3p}{\rho} + \frac{3p}{\rho^3} + \frac{3p}{\rho^5} + \dots \leq \frac{3p}{\rho} \sum_{k=0}^\infty \frac{1}{\rho^{2k}} = \frac{3p}{\rho}\frac{\rho^2}{\rho^2 -1} = \bigO{\frac{p}{\rho}}.
     \]
     Since the ratio of this value and $p$ is $\bigO{1/\rho}$, with probability $1 - \bigO{1/\rho}$ the sample returned is the heaviest element in $\Gamma_i$.
   \end{proof}

   \noindent Finally, we examine the case when $s$ comes from $\Lambda_i$:
   \begin{proposition}\label{prop:s:from:lambdai}
     Condition that $s_i$ comes from $\Lambda_i$.
     Then either:
     \begin{itemize}
       \item $|\Lambda_i \cap B_j| = 0$ for all $j \in [2r]$, in which case $\dtv(T_i^{\yes},T_i^{\vphantom{\yes}\no}) = 0$;
       \item or $|\Lambda_i \cap B_j| \neq 0$ for some $j \in [2r]$, in which case $\dtv(T_i^{\yes},T_i^{\vphantom{\yes}\no}) \leq \bigO{\gamma + \frac{1}{\varphi}} = \bigO{\frac{1}{q^{5/2}}}$
     \end{itemize}
   \end{proposition}
   \begin{proof}
     The former case follows from the definition of the sampling model -- since $\Lambda_i$ does not intersect any of the buckets, the sample will be labeled as such.
     Furthermore, the sample returned will be drawn uniformly at random from $\Lambda_i$, and the probability of each atom will be proportional to the cardinality of its intersection with $\Lambda_i$, in both the $\yes$ and the $\no$-instances.

     We next turn to the latter case. Let $\mathcal{X}$ be the event that, if the intersection of $\Lambda_i$ and some atom $A$ has a non-empty intersection with an odd number of buckets, then $s_i$ does not come from the unpaired bucket.
     Note that $\iha{i}$ and the definition of a large atom imply that an unpaired bucket can only occur if the atom intersects at least $\varphi$ bucket-pairs: conditioned on the sample coming from a particular atom, the probability that it comes from the unpaired bucket is $\bigO{{1}/{\varphi}}$.
     Since the choice of $A$ was arbitrary, we may remove the conditioning, and note that $\Pr(\mathcal{X}) = 1 - \bigO{{1}/{\varphi}}$.\medskip

     \noindent Since
     \begin{align}
       \dtv(T_i^{\yes},T_i^{\no}) 
          &\leq \dtv(T_i^{\yes},T_i^{\vphantom{\yes}\no} \mid \mathcal{X})\proba(\mathcal{X}) 
              + \dtv(T_i^{\yes},T_i^{\vphantom{\yes}\no} \mid\mathcal{\bar X})\proba(\mathcal{\bar X}) \notag\\
          &\leq \dtv(T_i^{\yes},T_i^{\vphantom{\yes}\no} \mid \mathcal{X})+ \bigO{1/\varphi},
     \end{align}
     it remains to show that $\dtv(T_i^{\yes},T_i^{\vphantom{\yes}\no} \mid \mathcal{X}) \leq \bigO{\gamma}$.\medskip

     First, we focus on the distribution over atoms, conditioned on $\mathcal{X}$.
     Let $N^{A}$ be the number of bucket-pairs with which $A$ intersects both buckets, \ie, conditioned on $\mathcal{X}$, the sample could come from $2N^A$ buckets,
     and let $N \eqdef \sum_{\atom\in\operatorname{At}(A_1,\dots, A_i)} N^{A}$.
     By $\iha{i}$, the maximum amount of probability mass that can be assigned to atom $A$ is $\frac{(1+\gamma){|S|}N^A/{rn}}{(1-\gamma){|S|}N/{rn}}$, and the minimum is $\frac{(1-\gamma){|S|}N^A/{rn}}{(1+\gamma){|S|}N/{rn}}$, so the total variation distance in the distribution incurred by this atom is at most $O({\gamma {N^A}/N})$.
     Summing over all atoms, we get the desired result of $O(\gamma)$.\medskip

     Finally, we bound the distance on the distribution over bucket-pairs, again conditioned on $\mathcal{X}$.
     By $\iha{i}$ only large atoms will contain non-zero probability mass, so condition on the sample coming from some large atom $A$.
     Let $N^A$ be the number of bucket-pairs with which $A$ intersects both buckets, \ie, conditioned on $\mathcal{X}$, the sample could come from $2N^A$ buckets.
     Using $\iha{i}$, the maximum amount of probability mass that can be assigned to any intersecting bucket-pair is ${(1+\gamma)\frac{|A|}{rn}}({(1-\gamma)\frac{|A|}{rn}N^A})^{-1}$, and the minimum is ${(1-\gamma)\frac{|A|}{rn}}({(1+\gamma)\frac{|A|}{rn}N^A})^{-1}$, so the total variation distance in the distribution incurred by this bucket-pair is at most $O({\gamma /N^A})$.
     Summing this difference over all $N^A$ bucket-pairs, we get $\frac{2\gamma}{1 - \gamma^2} = O(\gamma)$.
     Since the choice of large atom $A$ was arbitrary, we can remove the conditioning on the choice of atom.
     The statement follows by applying the union bound on the distribution over bucket-pairs and the distribution over atoms. This concludes the proof of~\expref{Proposition}{prop:s:from:lambdai}.
   \end{proof}
   We note that in both cases, the cost in total variation distance which is incurred is $O({\frac{1}{\rho} + \gamma + \frac{1}{\varphi}})$, which implies $\ihc{i}$~--~proving~\expref{Lemma}{lemma:inductioni} .
   \end{proof}
\noindent This concludes the proof of~\expref{Lemma}{lem:induction}.
  \end{proof}

\noindent With~\expref{Lemma}{lem:induction} in hand, the proof of the main theorem is straightforward.
  \begin{proofof}{\expref{Theorem}{theo:main:lb:equiv}}
    Conditioned on~\expref{Lemma}{lem:banned},~\expref{Lemma}{lem:induction} implies that the distribution over the leaves in a $\yes$-instance vs. a $\no$-instance is $o(1)$.
    Since an algorithm's choice to accept or reject depends deterministically on which leaf is reached, this bounds the difference between the conditional probability of reaching a leaf which accepts.
    Since~\expref{Lemma}{lem:banned} occurs with probability $1 - o(1)$, the difference between the unconditional probabilities is also $o(1)$.
  \end{proofof}

\section{Lower Bounds for Non-Adaptive Algorithms}\label{sec:lb-uniform}
In this section, we prove our lower bounds for non-adaptive support-size estimation and uniformity testing, restated here.
  \thmsuppnalb*
  \thmunifnalb*
\noindent These two theorems follow from the same argument, which we outline below before turning to the proof itself. (Note that we will, in this section, establish~\expref{Theorem}{theo:uniform:na:lb} for \emph{constant} $\eps$, \ie, $\eps=1/4$; before explaining in~\expref{Section}{ssec:lb:unif:general} how to derive the general statement as a corollary.)

\paragraph{Structure of the proof.} By Yao's Minimax Principle, we consider deterministic tests and study their performance over random distributions, chosen to be uniform over a random subset of carefully picked size. The proof of~\expref{Theorem}{theo:support:na:lb:beta} then proceeds in 3 steps: in~\expref{Lemma}{lem:wishful}, we first argue that all queries made by the deterministic tester will (with high probability over the choice of the support size $s$) behave very ``nicely'' with regard to $s$, \ie, not be concentrated around it. Then, we condition on this to bound the total variation distance between the sequence of samples obtained in the two cases we ``oppose,'' a random distribution from a family $\mathcal{\D}_1$ and the corresponding one from a family $\mathcal{\D}_2$. In~\expref{Lemma}{lem:supp:na:queries:small} we show that the part of total variation distance due to samples from the small queries is zero, except with probability $o(1)$ over the choice of $s$. Similarly,~\expref{Lemma}{lem:supp:na:queries:small} allows us to say (comparing both cases to a third ``reference'' case, a honest-to-goodness uniform distribution over the whole domain, and applying a triangle inequality) that the remaining part of the total variation distance due  to samples from the big queries is zero as well, except again with probability $o(1)$. Combining these three lets us to conclude by properties of the total variation distance, as (since the queries are non-adaptive) the distribution over the sequence of samples is a product distribution. (Moreover, applying~\expref{Lemma}{lem:supp:na:queries:small} as a stand-alone enables us, with little additional work, to obtain~\expref{Theorem}{theo:uniform:na:lb}, as our argument in particular implies distributions from $\mathcal{\D}_1$ are hard to distinguish from uniform.) \medskip

\paragraph{The families $\mathcal{\D}_1$ and $\mathcal{\D}_2$.}
Fix $\gamma \geq \sqrt{2}$ as in~\expref{Theorem}{theo:support:na:lb:beta}; writing $\beta\eqdef \gamma^2$, we define the set
\begin{equation}
  \mathcal{S} \eqdef \setOfSuchThat{ \beta^k n^{1/4} }{ 0\leq k \leq \frac{\log n}{2\log\beta}} =  \{ n^{1/4}, \beta n^{1/4}, \beta^2 n^{1/4},\dots, , n^{3/4} \}
\end{equation}
A \no-instance $(\D_1,\D_2)\in\mathcal{\D}_1\times\mathcal{\D}_2$ is then obtained by the following process:
\begin{itemize}
\item Draw $s$ uniformly at random from $\mathcal{S}$.
\item Pick a uniformly random set $S_1\subseteq[n]$ of size $s$, and set $\D_1$ to be uniform on $S_1$. 
\item Pick a uniformly random set $S_2\subseteq[n]$ of size $\beta s$, and set $\D_2$ to be uniform on $S_2$. 
\end{itemize}
(Similarly, a \yes-instance is obtained by first drawing a \no-instance $(\D_1,\D_2)$, and discarding $\D_2$ to keep only  $(\D_1,\D_1)\in\mathcal{\D}_1\times\mathcal{\D}_1$.)

We will argue that no algorithm can distinguish with high probability between the cases $(\D_1,\D_2)\sim \mathcal{\D}_1\times\mathcal{\D}_2$ and $(\D_1,\D_2)\sim \mathcal{\D}_1\times\mathcal{\D}_1$, by showing that in both cases $\D_1$ and $\D_2$ generate transcripts indistinguishable from those the \emph{uniform} distribution $\uniformOn{[n]}$ would yield. This will imply~\expref{Theorem}{theo:support:na:lb:beta}, as any algorithm to estimate the support within a multiplicative $\gamma$ would imply a distinguisher between instances of the form $(\D_1,\D_1)$ and $(\D_1,\D_2)$ (indeed, the support sizes of $\D_1$ and $\D_2$ differ by a factor $\beta = \gamma^2$). (As for~\expref{Theorem}{theo:uniform:na:lb}, observe that any distribution $\D_1\in\mathcal{\D}_1$ has constant distance from the uniform distribution on $[n]$, so that a uniformity tester must be able to tell $\D_1$ apart from $\uniformOn{[n]}$.)

\paragraph{Small and big query sets.}
Let \Tester be any deterministic non-adaptive algorithm making $q_\Tester\leq q=\frac{1}{80000}\frac{\log n}{\log^2 \beta}$ queries. Without loss of generality, we can assume \Tester makes exactly $q$ queries, 
and denote them by $A_1, \dots, A_q\subseteq[n]$. Moreover, we let
$a_i=\abs{A_i}$, and (again without loss of generality) write
$a_1\geq \dots\geq a_q$. 

As a preliminary observation, note that for any $A\subset[n]$ and $0\leq s\leq n$ we have
\[
\shortexpect_S \abs{S\cap A} = \frac{\abs{A}s}{n}
\]
where the expectation is over a uniform choice of $S$ among all $\binom{n}{s}$ subsets of size $s$. This observation will lead us to divide the query sets $A_i$ in two groups, depending on the expected size of their intersection with the (random) support.

With this in mind, the following definition will be crucial to our proof. Intuitively, it captures the distribution of \emph{sizes of intersection} of various query sets with the randomly chosen set $S$. 
\begin{definition}
Let $s \geq 1$, and $\mathcal{A}=\{a_1, \dots, a_q\}$ be any set of $q$ integers. For a real number $t>0$, define
\begin{equation}\label{def:thitpoints}
  C_t(s)\eqdef \abs{ \setOfSuchThat{ i \in [q] }{ \frac{a_i s}{n}\in\left({\beta^{-t}}, \beta^t\right) } }
\end{equation}
to be the number of \emph{$t$-hit points of $\mathcal{A}$ (for $s$)}. 
\end{definition}
\noindent The next result will be crucial to prove our lower bounds: it roughly states that if we consider
the $a_i$  
and scale them by the random quantity $s/n$, then the distribution of the random variable generated has an exponential tail with respect to $t$. 
\begin{lemma}[Hitting Lemma]\label{lem:wishful}
  Fix $\mathcal{A}$ as in the previous definition. If $s$ is drawn uniformly at random from $\mathcal{S}$, then with probability at least $99/100$.
  \begin{align}\label{eqn:ct}
    \sup_{t>0}\frac{C_t(s)}{t}<\frac{2}{100}.
  \end{align}
\end{lemma}
\begin{proof}
Without loss of generality, assume that the $a_i$ are in decreasing order.
We will work in the logarithmic domain: a number $a_i$ contributes to
$C_t(s)$ if and only if $\log s \in[\log
(n/a_i)- t\log \beta,\log (n/a_i)+ t\log \beta]$.
Indeed, we can restate the lemma in an additive form. Let 
$\mathcal{A}=\{\alpha_1, \dots, \alpha_m\}$ be any set of numbers in $[0,\log n]$. 
These are defined as transformations of $a_i$'s: $\alpha_i \triangleq \log(n/a_i)$.
In the additive restating, $x$ will play the role of $\log s$, or equivalently, $s = 2^x$.
For a point $x\in[0,\log n]$, let $\ell_j^x$ and $r_j^x$ denote the distance of $x$
from the $j$th point to its left and right, respectively, from the set
$\mathcal{A}$. 
More precisely, if $\alpha_\gamma \leq x \leq \alpha_{\gamma+1}$, $\ell_j^x \triangleq x - \alpha_{\gamma+1-j}$.
If we consider only points to the left of $x$, $\ell_j^x/\log \beta$ is the least value of $t$ such that $C_{t}(2^x) = j$.
Therefore, if $t^x_j \triangleq \frac{1}{\log \beta} \min \left\{\ell_j^x, r_j^x\right\}$, then we are guaranteed that $j \leq C_{t^x_j}(2^x) \leq 2j$.

Observe that $C_t(2^x)$ is a piecewise-constant function which is monotone non-decreasing in $t$.
Therefore, the supremum of $\frac{C_t(2^x)}{t}$ is attained at one of these discontinuities:
\[
\sup_{t > 0} \frac{C_t(2^x)}{t} = \max_{\{t\ :\ \forall \delta > 0, C_{t-\delta}(2^x) < C_t(2^x)\}} \frac{C_t(2^x)}{t}.
\] 
We can in turn upper bound this by looking at the set of all $t^x_j$.
Note that this may ignore some discontinuous points: for instance, suppose that $\ell_j^x < r_j^x$, then $r_j^x/\log \beta$ will not be considered.
However, we note that in this case, $C_{r_j^x/\log \beta}(2^x) \leq 2C_{\ell_j^x/\log \beta}(2^x) = 2C_{t^x_j}(2^x)$:
\[
\frac{C_{r_j^x/\log \beta}(2^x)}{r_j^x/\log \beta} \leq \frac{C_{r_j^x/\log \beta}(2^x)}{\ell_j^x/\log \beta} \leq \frac{2C_{\ell_j^x/\log \beta}(2^x)}{\ell_j^x/\log \beta} = \frac{2C_{t_j^x}(2^x)}{t_j^x}.
\]
Therefore, 
\[
\sup_{t > 0} \frac{C_t(2^x)}{t} \leq \max_{\{t_j^x\ :\ j \in [q]\}}\frac{2C_{t_j^x}(2^x)}{t_j^x} \leq  \max_{\{t_j^x\ :\ j \in [q]\}}\frac{4j}{t_j^x} = \max_{\{t_j^x\ :\ j \in [q]\}}\frac{4j\log \beta}{\min\{\ell_j^x, r_j^x\}}
\]
We would like to upper bound this term by $2/100$.
Equivalently, we satisfy the lemma conditions if
\[
\min_j \left\{\frac{t_j^x}{j}\right\} \geq 200 \log \beta.
\]

For a constant $c>0$, let $S_c$ be the set of all points $x$ (where recall $s = 2^x$ is selected according to the distribution $\mathcal{S}$) such that violate this inequality for $c$:
\[
\min_j \left\{\frac{t_j^x}{j}\right\} \leq c.
\]
We would like to upper bound the probability that a randomly selected $x$ violates this inequality for $c = 200 \log \beta$ by $1/100$.
Equivalently, since $x$ is selected uniformly at random from $\frac{\log n}{2 \log \beta}$ different values, we would like to upper bound the size of $S_{200 \log \beta}$ by $\log n/200 \log \beta$.
We do this with the following claim.
\begin{claim}\label{claim:hitting:again}
$|S_c|\leq 2cq$.
\end{claim} 
Substituting in $c = 200 \log \beta$ and $q = \log n / 80000 \log^2 \beta$ will give the desired result.

\begin{proofof}{\expref{Claim}{claim:hitting:again}}
We consider the set of points in $S_{c,\ell}\subset S_c$ that satisfy
${\ell_j^x}/j<c$ for some $j$, and show that their measure is at
most $cq$. An identical bound holds for the set of points of $S_c$ for
which ${r_j^x}/j<c$. Let $S_{c,\ell}^{i}\subset S_{c,\ell}$ be
the set of points in $S_{c,\ell}$ that satisfy $\min_j \left\{\frac{t_j^x}{j}\right\}<c$ with respect to
the set $\{\alpha_1,\dots, \alpha_i\}$. We will show by induction that 
$
|S_{c,\ell}^{i}|<ci. 
$\smallskip

For the first point $\alpha_1$, the set $S_{c,\ell}^1=[\alpha_1,
\alpha_1+c]$. Suppose by induction that $|S_{c,\ell}^{i}|<ci.$ Let $x_i$ be
the right-most point in the set $S_{c,\ell}^{i}$. Then it is clear
that $x_i>\alpha_i$, in fact $x_i\geq \alpha_i+c$. Furthermore, either $x_i= \log n$,
or ${\ell_j^{x_i}}/j=c$ for some $j$. Moreover, we claim that $[\alpha_i,x_i]\in
S_{c,\ell}^i$. Indeed, for the same $j$ that ${\ell_j^{x_i}}/j<c$,
all points in $[\alpha_i,x_i]$ satisfy the condition. If $x_i=\log n$, then the
result holds trivially. We therefore consider the  point $\alpha_{i+1}$ and
prove the inductive step for $x_{i}<\log n$.  There are two cases:
\begin{description}
  \item[If {${\alpha_{i+1}\geq x_i}$:}] In this case, $S_{c,\ell}^{i+1}=S_{c,\ell}^i\cup
[\alpha_{i+1},x_{i+1}]$. We have to show that $x_{i+1}\leq \alpha_{i+1}+c$.  
Suppose to the contrary that $x_{i+1}>\alpha_{i+1}+c\geq x_i+c$. Then there is a
point $\alpha_h$ for $h\leq i$, such that 
$
\frac{x_{i+1}-\alpha_{h}}{i+2-h}<c, 
$
and then 
$
\frac{\alpha_{i+1}+c-\alpha_{h}}{i+2-h}<c, 
$
so that 
\[
\frac{\alpha_{i+1}-\alpha_{h}}{i+1-h}<c, 
\]
however, this implies that $\alpha_{i+1}\in S_{c,\ell}^i$, contradicting
the assumption of this case. 

\item[If {$\alpha_{i+1}< x_i$:}]  In this case, $S_{c,\ell}^{i+1}=S_{c,\ell}^i\cup
[x_{i},x_{i+1}]$. We have to show that $x_{i+1}\leq x_{i}+c$. Suppose
on the contrary that $x_{i+1}> x_{i}+c>\alpha_{i+1}+c$. Let $h$ be the
index such that 
$
\frac{x_{i+1}-\alpha_{h}}{i+2-h}<c, 
$
and therefore, 
$
\frac{x_{i}+c-\alpha_{h}}{i+2-h}<c, 
$
implying that 
\[
\frac{x_{i}-\alpha_{h}}{i+1-h}<c, 
\]
contradicting that $x_i$ is the rightmost point of $S_{c,\ell}^i$.
\end{description}  
\end{proofof}
This concludes the proof of the hitting lemma.
\end{proof}

We proceed to show how to use this lemma to bound the contribution of various types of queries to the distinguishability of $\D_1$ and $\D_2$.
In particular, we will apply~\expref{Lemma}{lem:wishful} to the set of query sizes $\{a_1,\dots,a_q\}$.

\noindent Recall that 
the $a_i$  
are non-increasing. If $a_{q^\prime}s/n\leq 1$ let $q^\prime\eqdef q+1$, otherwise define
$q^\prime$ as the largest integer such that $a_{q^\prime}s/n>1$. We partition the queries made by \Tester in two: $A_{1},\dots,A_{q^\prime}$ are said to be \emph{big}, while  $A_{q^\prime+1},\dots,A_{q}$ are \emph{small queries}.

\begin{lemma}\label{lem:supp:na:queries:small}
With probability at least $1-2^{-10}$, a random distribution from
$\mathcal{\D}_1$ or from $\mathcal{\D}_2$ does 
not intersect with any small query. \end{lemma}
\begin{proof}
    Let $s$ be the random size drawn for the definition of the instances. We first claim that $\shortexpect[\abs{ A_{q^\prime+j}\cap S}] \leq \beta^{-50j}$ for all $j\geq 1$, where the expectation is over the uniform choice of set $S_1$ for $\D_1$. Indeed, by contradiction suppose there is a $j\geq 1$ such that 
    $\shortexpect[\abs{A_{q^\prime+j}\cap S}]=\frac{a_{q^\prime+j}s}n > \beta^{-50j}$. 
    By definition of $q^\prime$, for $1\leq i\leq j$, 
    \[
    1\ge \frac{a_{q^\prime+i}s}n>\beta^{-50j}. 
    \]
    Therefore, the queries $A_{q^\prime}, A_{q^\prime+1}, \dots,
    A_{q^\prime+j}$ contribute to $C_{50j}$, and we obtain 
    $
    \frac{C_{50j}}{50j}\geq \frac{j}{50j}=\frac{2}{100},
    $
    contradicting~\expref{Lemma}{lem:wishful}. Thus, the expected
    intersection can be bounded as follows:
    \begin{align*}
      \shortexpect[\abs{(A_{q^\prime+1}\cup A_{q^\prime+2}\dots\cup A_{q})\cap S}]
      &\leq \shortexpect[\abs{A_{q^\prime+1}\cap S}]+ \shortexpect[\abs{A_{q^\prime+2}\cap S}]+\dots+\shortexpect[\abs{A_q\cap S}]\\
      &\leq \beta^{-50}+\beta^{-100}+\dots\\
      &\leq 2^{-12},
    \end{align*}
    since $\beta\geq 2$. 
    From this, we obtain the result holds for $\mathcal{\D}_1$ by Markov's inequality. 
    The same applies to $\mathcal{\D}_2$ with probability of intersection at
    most $2^{-10}$, proving the lemma. 
\end{proof}

We now turn our attention to the sets with \emph{large} intersections. We will show that under $\mathcal{\D}_1$ and
$\mathcal{\D}_2$, the results of querying the sets $A_1,\dots A_{q^\prime}$ are indistinguishable from simply picking
elements uniformly from the sets $A_1, \dots, A_{q^\prime}$. More precisely, we establish the following.
\begin{lemma}\label{lem:supp:na:queries:big}
Let $\eta^\ast=2^{-10}$ and $\eta_s=1/100$; and fix $\ell\in\{1,2\}$. At least an $1-\eta_s$ fraction of elements 
$s_1,\dots,s_{q^\prime}\in A_1\times A_2, \dots,A_{q^\prime}$ satisfy
\begin{equation}
\probaDistrOf{\ell}{(s_1,\dots, s_{q^\prime})}\in\left[ 1-5\eta^\ast, 1+5\eta^\ast\right]\cdot\frac{1}{\abs{A_1}\dots\abs{A_{q^\prime}}}, 
\end{equation}
where $\probaDistrOf{\ell}{(s_1,\dots, s_{q^\prime})}$ denotes the probability that $s_1\dots s_{q^\prime}$ are the results of the queries $A_1,\dots,A_{q^\prime}$ under $\COND_{\D_\ell}$.
\end{lemma}

As this holds for most distributions in both $\mathcal{\D}_1$ and $\mathcal{\D}_2$, this implies the queries are indistinguishable from the
product distribution over $ A_1\times A_2, \dots, A_{q^\prime}$ (which is the one induced by the same queries on the uniform distribution over $[n]$) in either case, with probability at least $1-\eta_s-5\eta^\ast$.

\begin{proofof}{\expref{Lemma}{lem:supp:na:queries:big}}
  From standard Chernoff-like 
  concentration bounds for hypergeometric random variables (Lemma~\ref{lem:chernoff-hyp}),  we obtain the claim below.
  \begin{claim}\label{claim:concent_inter}
  Suppose $A$ is a set of size $a$, and $S$ is a uniformly chosen random
  set of size $s$. Then, for all $\eta\in(0,1]$, we have 
  $
  \probaOf{\abs{A\cap S} >(1+\eta)\frac{as}{n} } < e^{-\eta^2\cdot \frac{as}{3n}}
  $
  and
  $
  \probaOf{\abs{A\cap S}<(1-\eta)\frac{as}{n} } < e^{-\eta^2\cdot \frac{as}{3n}}
  $.
  \end{claim}
  We use this to show that indeed all the $\abs{A_i\cap S}$  concentrate around their expected values for $1 \leq j\leq q^{\prime}$.
 First note that, as a consequence of~\expref{Lemma}{lem:wishful}, it is the case that these expected values satisfy $a_{q^\prime-j}s/n\ge \beta^{50(j+1)}$ for every $0\leq j\leq q^\prime-1$ (with probability at least $99/100$). Conditioning on this, we first invoke~\expref{Claim}{claim:concent_inter} on $A_j$ with $\eta=3\cdot\beta^{20(j+1)}$, and then apply a union bound to obtain
  \begin{equation}\label{claim:concent_inter-new}
  \probaOf{\exists j\in [q^\prime] \text{ s.t. } \abs{A_j\cap S} \notin \Big[1-4\cdot\beta^{-20(j+1)}, 1+4\cdot\beta^{-20(j+1)} \Big] \cdot \frac{a_js}{n} } <
  e^{-\beta^{10}}
  \end{equation}
  \noindent \ie, with high probability all intersections simultaneously concentrate around their expected values. \smallskip

  \noindent Note that since $s$ is at most $n^{3/4}$, each $A_i$ under
  consideration has size at least $n\beta^{50}/n^{3/4}>n^{1/4}$. Therefore, the probability that a random
  selection of elements from $A_1,\dots, A_{q^\prime}$ exhibits no collision is at least
  \[
  \prod_{i=1}^{q^\prime}\frac{|A_i|-q^\prime}{|A_i|}\ge
  \left(1-\frac{q^\prime}{n^{1/4}}\right)^{q^\prime} \ge
  1-\frac{(q^{\prime})^2}{n^{1/4}}>1-\frac{\log^2 n}{n^{1/4}}.
  \]
  We henceforth condition on this event.

  Let $N={n\choose s}$ be the number of outcomes for the set $S$. 
  We write $N_0\ge N(1-e^{-\beta^{10}})$ for the number of such sets for which~\expeqref{equation}{claim:concent_inter-new} holds. Let $s_1^{q^\prime}$ denote
  $s_1\dots s_{q^\prime}$.  
  For a set of distinct 
  $(s_1,\dots, s_{q^\prime})\in A_1\times\dots \times A_{q^\prime}$, let
  $N(s_1^{q^\prime})={n-q^\prime\choose s-q^\prime}$ be the number of
  sets of size $s$ that contain $s_1^{q^\prime}$, and let $N_0(s_1^{q^\prime})$ of
  them satisfy~\expeqref{equation}{claim:concent_inter-new}. 

  By Markov's inequality, with probability at least $1-e^{-\beta^{9}}$, for a
  randomly chosen $s_1^{q^\prime}$ we have
  $N_0(s_1^{q^\prime})/N(s_1^{q^\prime})>1-e^{2-\beta^{9}}$.  For any such
  $s_1^{q^\prime}$, 
  \begin{align*}
    \probaOf{s_1^{q^\prime}}&\ge  \frac{N_0(s_1^{q^\prime})}{N}\cdot\prod_{i=1}^{q^\prime} \frac{1}{\abs{A_i\cap S}}
    \geq
    (1-e^{2-\beta^{9}})\frac{s(s-1)\dots(s-q^\prime+1)}{n(n-1)\dots(n-q^\prime+1)}
    \cdot (1-4\cdot \beta^{-19}) \prod_{i=1}^{q^\prime}\frac n{a_i s}\\
    &\geq (1-6\cdot \beta^{-19}) \prod_{i=1}^{q^\prime}\frac 1{a_i},
  \end{align*}
  for large $n$ and as $\abs{S} > n^{1/4}$.
  Since the sum of probabilities of elements is at most 1, the other
  side of the inequality in~\expref{Lemma}{lem:supp:na:queries:big} follows. 
\end{proofof}

\begin{proofof}{\expref{Theorem}{theo:support:na:lb:beta} and~\expref{Theorem}{theo:uniform:na:lb}}
Let $T_1$ (resp. $T_2$, $T_U$) be the distribution over transcripts generated by the queries $A_1,\dots,A_q$ when given conditional access to the distribution $\D_1$ from a $\no$-instance (resp. $\D_2$, resp. uniform $\uniformOn{[n]}$); that is, a distribution over $q$-tuples in $A_1\times\dots\times A_q$. Since the queries were non-adaptive, we can break $T_1$ (and similarly for $T_2$, $T_U$) in $T^{\rm big}_1\times T^{\rm small}_1$ according to $q^\prime$, and use~\expref{Lemma}{lem:supp:na:queries:big} and~\expref{Lemma}{lem:supp:na:queries:small} separately to obtain
$\totalvardist{T_1}{T_2} \leq \eta_s+\eta^\ast+ 2^{-10} < 1/50$ and $\totalvardist{T_1}{T_U} \leq \eta_s+\eta^\ast+ 2^{-10} < 1/50$ (for the latter, recalling that queries that do not intersect the support receive samples exactly uniformly distributed in the query set) -- thus establishing both theorems.
\end{proofof}

\subsection{On the dependence on $\eps$ in~\expref{Theorem}{theo:uniform:na:lb}}\label{ssec:lb:unif:general}
We remark that~\expref{Theorem}{theo:uniform:na:lb}, by establishing a lower bound of $\bigOmega{\log n}$ queries for non-adaptive testing of uniformity with constant distance parameter $1/4$, immediately implies, by a standard argument, an $\bigOmega{(\log n)/\eps}$ lower bound for distance parameter $\eps\in(0,1/4)$. In more detail, this is a consequence of the following reduction: any $q(n,\eps)$-query non-adaptive tester for uniformity $\Tester$ can be used, given conditional access to some distribution $\D$ on $[n]$, on the mixture distribution
\begin{equation}
    \D_\eps \eqdef 4\eps \D + (1-4\eps)\uniformOn{[n]}\,,
\end{equation}
for which a conditional oracle can be easily simulated given a conditional oracle for $\D$. Moreover, answering $q(n,\eps)$ to $\D_\eps$ can be done with an expected $4\eps q(n,\eps)$ conditional queries to $\D$. As it is immediate to see that $\dtv(\D_\eps, \uniformOn{[n]} ) = 4\eps \dtv( \D, \uniformOn{[n]} )$, we get that $\Tester$ can be used to obtain a tester for non-adaptive testing of uniformity with constant distance parameter $1/4$, with query complexity $O(\eps q(n,\eps))$ for every $\eps < 1/4$ (converting the expected query complexity to a worst-case one is straightforward via Markov's inequality followed by success probability amplification by a constant number of repetitions). Therefore, the lower bound of~\expref{Theorem}{theo:uniform:na:lb} implies that $q(n,\eps) = \bigOmega{(\log n)/\eps},$ as claimed. 
\noindent It is also worth noting that the above argument does \emph{not} yield an analogue statement for support-size estimation via~\expref{Theorem}{theo:support:na:lb:beta}. Indeed, mixing the distribution $\D$ with the uniform distribution does not preserve the support size in that case (nor the guarantee that every point of the support has probability mass at least $\tau/n$).

\section{An Upper Bound for Support-Size Estimation}\label{sec:ub-supp-size}
  In this section, we prove our upper bound for constant-factor support-size estimation, reproduced below.
  \thmsuppsize*

Before describing and analyzing our algorithm, we shall need the following results, that we will use as subroutines: the first one will help us detecting when the support is already dense. The second, assuming the support is sparse enough, will enable us to find an element with zero probability mass, which can afterwards be used as a ``reference'' to verify whether any given element is inside or outside the support. Finally, the last one will use such a reference point to check whether a candidate support size $\sigma$ is smaller or significantly bigger than the actual support size.

\begin{lemma}\label{lemma:testsmallsupport}
  Given $\tau > 0$ and \COND access to a distribution $\D$ such that each support element has probability at least $\tau/n$, as well as parameters $\eps\in(0,1/2), \delta\in(0,1)$, there exists an algorithm  \textsc{TestSmallSupport} (\expref{Algorithm}{algo:distinguish:epssmall:support}) that makes $\tildeO{1/(\tau\eps^2)+1/\tau^2}\cdot\log(1/\delta)$ queries to the oracle, and satisfies the following. \textsf{(i)} If $\supp{\D} \geq (1-\eps/2)n$, then it returns \accept with probability at least $1-\delta$; \textsf{(ii)} if $\supp{\D} \leq (1-\eps)n$, then it returns \reject with probability at least $1-\delta$.
\end{lemma}

\begin{lemma}\label{lemma:getnonsupport}
  Given \COND access to a distribution $\D$, an upper bound $m < n$ on $\supp{\D}$, as well as parameter $\delta\in(0,1)$, there exists an algorithm  \textsc{GetNonSupport} (\expref{Algorithm}{algo:get:non:support}) that makes $\tildeO{\log^2\frac{1}{\delta}\log^{-2}\frac{n}{m}}$ queries to the oracle, and returns an element $r\in\domain$ such that $r\notin\supp{\D}$ with probability at least $1-\delta$.
\end{lemma}

\begin{lemma}\label{lemma:isatmostsupportsize}
  Given \COND access to a distribution $\D$, inputs $\sigma \geq 2$ and $r\notin\supp{\D}$, as well as parameters $\eps\in(0,1/2), \delta\in(0,1)$, there exists an algorithm  \textsc{IsAtMostSupportSize} (\expref{Algorithm}{algo:iamss}) that makes $\tildeO{1/\eps^2}\log(1/\delta)$ queries to the oracle, and satisfies the following. The algorithm returns either \yes or \no, and \textsf{(i)} if $\sigma \geq \supp{\D}$, then it returns \yes with probability at least $1-\delta$; \textsf{(ii)} if $\sigma > (1+\eps)\supp{\D}$, then it returns \no with probability at least $1-\delta$.
\end{lemma}
We defer the proofs of these 3 lemmata to the next subsections, and now turn to the proof of the theorem. 

\begin{proof}
The algorithm is given in~\expref{Algorithm}{algo:estimate:support}, and at a high-level works as follows: if first checks whether the support size is big (an $1-\bigO{\eps}$ fraction of the domain), in which case it can already stop and return a good estimate. If this is not the case, however, then the support is sparse enough to efficiently find an element $r$ \emph{outside} the support, by taking a few uniform points, comparing and ordering them by probability mass (and keeping the lightest). This element $r$ can then be used  as a reference point in a (doubly exponential) search for a good estimate: for each guess $\tilde{\omega}$, a random subset $S$ of size roughly $\tilde{\omega}$ is taken, a point $x$ is drawn from $\D_S$, and $x$ is compared to $r$ to check if $\D(x) > 0$. If so, then $S$ intersects the support, meaning that $\tilde{\omega}$ is an upper bound on $\omega$; repeating until this is no longer the case results in an accurate estimate of $\omega$.
\begin{algorithm}[H]
  \begin{algorithmic}[1]
        \If{$\textsc{TestSmallSupport}_\D(\eps,\frac{1}{10})$ returns \accept}
      \Return $\tilde{\omega} \gets (1-\eps^2)n$ \label{algo:es:step:firstcheck}
    \EndIf
        \State Call $\textsc{GetNonSupport}_\D((1-\frac{\eps}{2})n,\frac{1}{10})$ to obtain a non-support reference point $r$. \label{algo:es:step:gns}
    \For{$j$ \textbf{from} $0$ \textbf{to} $\log_{1+\eps}\log_{1+\eps} n$}
      \State Set $\tilde{\omega}\gets (1+\eps)^{(1+\eps)^j}$.
      \State Call $\textsc{IsAtMostSupportSize}_\D(\tilde{\omega}, r, \eps, \frac{1}{100\cdot (j+1)^2})$ to check if $\tilde{\omega}$ is an upper bound on $\omega$. \label{algo:es:step:ialss:1}
      \If{ the call returned \no }
          \State Perform a binary search on $\{(1+\eps)^{j-1},\dots, (1+\eps)^j\}$ to find $i^\ast$, the smallest $i \geq 2$ such that $\textsc{IsAtMostSupportSize}_\D((1+\eps)^{i}, r, \eps, \frac{1}{10 (j+1)})$ returns \no. \label{algo:es:step:ialss:2}
          \State \Return  $\tilde{\omega}\gets (1+\eps)^{i^\ast-1}$.
      \EndIf
    \EndFor
  \end{algorithmic}
  \caption{\label{algo:estimate:support}$\textsc{EstimateSupport}_\D$}
\end{algorithm}

In the rest of this section, we formalize and rigorously argue the above. Conditioning on each of the calls to the subroutines \textsc{TestSmallSupport}, \textsc{GetNonSupport} and \textsc{IsAtMostSupportSize} being correct (which overall happens except with probability at most $1/10+1/10+\sum_{j=1}^\infty 1/(100j^2) + 1/10 < 1/3$ by a union bound), we show that the output $\tilde{\omega}$ of \textsc{EstimateSupport} is indeed within a factor $(1+\eps)$ of $\omega$.
\begin{itemize}
  \item If the test on Step~\ref{algo:es:step:firstcheck} passes, then by~\expref{Lemma}{lemma:testsmallsupport} we must have $\supp{\D} > (1-\eps)n$. Thus, the estimate we return is correct, as $[(1-\eps)n,n]\subseteq[\tilde{\omega}/(1+\eps), (1+\eps)\tilde{\omega}]$.
  \item Otherwise, if it does not then by~\expref{Lemma}{lemma:testsmallsupport} it must be the case that $\supp{\D} < (1-\eps/2)n$.
\end{itemize}
Therefore, if we reach Step~\ref{algo:es:step:gns} then $(1-\eps/2)n$ is indeed an upper bound on $\omega$, and $\textsc{GetNonSupport}$ will return a point $r\notin\supp{\D}$ as expected. The analysis of the rest of the algorithm is straightforward: from the guarantee of \textsc{IsAtMostSupportSize}, the binary search will be performed for the first index $j$ such that $\omega\in[(1+\eps)^{(1+\eps)^{j-1}}, (1+\eps)^{(1+\eps)^{j}}]$; and will be on a set of $(1+\eps)^{j-1}$ values. Similarly, for the value $i^\ast$ eventually obtained, it must be the case that $(1+\eps)^{i^\ast} > \omega$ (by contrapositive, as \no was returned by the subroutine) but $(1+\eps)^{i^\ast-1} \leq (1+\eps)\omega$ (again, as the subroutine returned \yes). But then, $\tilde{\omega}=(1+\eps)^{i^\ast-1}\in(\omega/(1+\eps),(1+\eps)\omega]$ as claimed.

\paragraph{Query complexity.} The query complexity of our algorithm originates from the following different steps:
\begin{itemize}
  \item the call to \textsc{TestSmallSupport}, which from~\expref{Lemma}{lemma:testsmallsupport} costs $\tildeO{1/\eps^2}$ queries;
  \item the call to \textsc{GetNonSupport}, on Step~\ref{algo:es:step:gns}, that from the choice of the upper bound also costs $\tildeO{1/\eps^2}$ queries;
  \item the (at most) $\log_{1+\eps}\log_{1+\eps} n = \bigO{(\log\log n)/\eps}$ calls to \textsc{IsAtMostSupportSize} on Step~\ref{algo:es:step:ialss:1}. Observing that the query complexity of \textsc{IsAtMostSupportSize} is only $\tildeO{1/\eps^2}\cdot\log(1/\delta)$, and from the choice of $\delta=\frac{1}{(j+1)^2}$ at the $j$-th iteration this step costs at most
  \[
      \tildeO{\frac{1}{\eps^2}}\cdot\sum_{j=1}^{\log_{1+\eps}\log_{1+\eps} n} \bigO{\log(j^2)} = \tildeO{\frac{1}{\eps^2}\log_{1+\eps}\log_{1+\eps} n}= \tildeO{\frac{1}{\eps^3}\log_{1+\eps}\log_{1+\eps} n}
  \] 
  queries.
  \item Similarly, Step~\ref{algo:es:step:ialss:2} results in at most $j\leq {\log\log n}$ calls to \textsc{IsAtMostSupportSize} with $\delta$ set to ${1/(10(j+1))}$, again costing $\tildeO{\frac{1}{\eps^2}}\cdot\log j = \tildeO{\frac{1}{\eps^2}\log_{1+\eps}\log_{1+\eps} n} = \tildeO{\frac{1}{\eps^3}\log \log n} $ queries.
\end{itemize}
Gathering all terms, the overall query complexity is $\tildeO{\frac{\log\log n}{\eps^3}}$, as claimed.
\end{proof}

\subsection{Proof of~\expref{Lemma}{lemma:testsmallsupport}}
Hereafter, we assume without loss of generality that $\tau < 2$: indeed, if $\tau \geq 2$ then the support is of size at most $n/2$, and it suffices to return \reject to meet the requirements of the lemma. We will rely on the (easy) fact below, which ensures that any distribution with dense support and minimum non-zero probability $\tau/n$ put significant mass on ``light'' elements.
\begin{fact}\label{fact:big:mass:on:small:mass}
Fix any $\eps\in[0,1)$. Assume $\D$ satisfies both $\supp{\D} \geq (1-\eps)n$ and $\D(x)\geq \tau/n$ for $x\in\supp{\D}$. Then, setting $L_\eps\eqdef\setOfSuchThat{ x\in\domain }{ \D(x)\in [\tau/n,2/n]}$, we have $\abs{L_\eps}\geq (1/2-\eps)n$ and $\D(L_\eps)\geq (1/2-\eps)\tau$.
\end{fact}
\begin{proof}
  As the second claim follows directly from the first and the minimum mass of elements of $L_\eps$, it suffices to prove that $\abs{L_\eps}\geq (1/2-\eps)n$. This follows from observing that
  \[
    1=\D(\domain) \geq \D(\domain\setminus L_\eps) \geq (\abs{\supp{\D}} - \abs{L_\eps})\frac{2}{n} \geq 2(1-\eps) - \frac{2\abs{L_\eps}}{n}
  \]
  and rearranging the terms.
\end{proof}

\paragraph{Description and intuition.} The algorithm (as described in~\expref{Algorithm}{algo:distinguish:epssmall:support}) works as follows: it first takes enough uniformly distributed samples $s_1,\dots,s_m$ to get (with high probability) an accurate enough fraction of them falling in the support to distinguish between the two cases. The issue is now to detect 
those $m_j$      
which indeed are support elements; note that we do not care about underestimating this fraction in case $\textsf{(b)}$ (when the support is at most $(1-\eps)n$, but importantly do not want to underestimate it in case  $\textsf{(a)}$ (when the support size is at least $(1-\eps/2)n$).
To perform this detection, we take constantly many samples \emph{according to $\D$} (which are therefore ensured to be in the support), and use pairwise conditional queries to sort them by increasing probability mass (up to approximation imprecision), and keep only the lightest of them, $t$. In case $\textsf{(a)}$, we now from~\expref{Fact}{fact:big:mass:on:small:mass} that with high probability our $t$ has mass in $[1/n,2/n]$, and will therefore be either much lighter than or comparable to \emph{any} support element: this will ensure that in case $\textsf{(a)}$ we do detect all of
the $m_j$   
that are in the support.

This also works in case $\textsf{(b)}$, even though~\expref{Fact}{fact:big:mass:on:small:mass} does not give us any guarantee on the mass of $t$. Indeed, either $t$ turns out to be light (and then the same argument
ensures that  
our estimate of the number of ``support''  
elements $m_j$    
is good), or $t$ is too heavy -- and then our estimate will end up being smaller than the true value. But this is fine, as the latter
only means we will reject the distribution (as we should, since we are in
the small-support case).

\begin{algorithm}[h]
  \begin{algorithmic}[1]
    \Require \COND access to $\D$; accuracy parameter $\eps\in(0,1/2)$, threshold $\tau > 0$, probability of failure $\delta$
    \State Repeat the following $\bigO{\log(1/\delta)}$ times and return the majority vote.
    \Loop
    \State Draw $m\eqdef \bigTheta{\frac{1}{\eps^2}}$ independent samples $s_1,\dots,s_m\sim\uniformOn{[n]}$.
    \State Draw $k\eqdef \bigTheta{\frac{1}{\tau}}$ independent samples $t_1,\dots,t_k\sim\D$.
  \ForAll{ $1\leq i < j\leq k$ } 
                                  \Comment{Order the $t_j$}   
      \State Call $\textsc{Compare}(\{t_i\},\{t_j\},\eta=\frac{1}{2},K=2,\frac{1}{4k^2})$ to get a $2$-approx. $\rho$ of $\frac{\D(t_j)}{\D(t_i)}$, \textsf{High} or \textsf{Low}. \label{algo:tss:step:compare} 
      \If{ $\textsc{Compare}$ returned \textsf{High} or a value $\rho$ }
        \State Record $t_i \preceq t_j$
      \Else
        \State Record $t_j \prec t_j$
      \EndIf
    \EndFor
    \State Set $t$ to be (any of the) smallest
 elements $t_j$    
    according to $\preceq$. 
    \ForAll{ $1\leq j \leq m$ } \Comment{Find the fraction of support elements among
the $m_j$}       
      \State Call $\textsc{Compare}(\{t\},\{s_j\},\eta=\frac{1}{2},K=\frac{2}{\tau},\frac{1}{4m})$ to get either a value $\rho$, \textsf{High} or \textsf{Low}. \label{algo:gns:step:compare:2}       \If{ $\textsc{Compare}$ returned \textsf{High} or a value $\rho \geq 1/2$ }
        \State Record $s_j$ as ``support.''
      \EndIf
    \EndFor 
    \If{ the number of
elements $s_j$   
      marked ``support'' is at least $(1-\frac{3}{4}\eps)m$}
      \Return \accept \label{algo:gns:step:accept}
    \Else\ \Return \reject
    \EndIf
    \EndLoop
  \end{algorithmic}
  \caption{\label{algo:distinguish:epssmall:support}$\textsc{TestSmallSupport}_\D$}
\end{algorithm}

\paragraph{Correctness.} Let $\eta$ be the fraction of the
elements $s_j$   
that are in the support of the distribution. By a multiplicative Chernoff bound and a suitable constant in our choice of $m$, we get that \textsf{(i)} if $\supp{\D} \geq 1-\eps/2$, then $\probaOf{\eta < 1-3\eps/4} \leq 1/12$, while  \textsf{(ii)} if $\supp{\D} \leq 1-\eps/2$, then $\probaOf{\eta \geq 1-3\eps/4} \leq 1/12$. We hereafter condition on this (\ie, $\eta$ being a good enough estimate). We also condition on all calls to $\textsc{Compare}$ yielding results as per specified, which by a union bound overall happens except with probability $1/12+1/12 = 1/6$, and break the rest of the analysis in two cases.

\begin{enumerate}[(a)]
\item Since the support size $\omega$ is in this case at least $(1-\eps/2)n$, from~\expref{Fact}{fact:big:mass:on:small:mass} we get that $\D(L_{\eps/2})\geq \frac{1-\eps}{2}\tau \geq \frac{\tau}{4}$. Therefore, except with probability at most $(1-\tau/4)^k < 1/12$, at least one of
the $t_j$        
will belong to $L_{\eps/2}$. When this happens, and by the choice of parameters in the calls to $\textsc{Compare}$, we get that $t\in L_{\eps/2}$; that is $\D(t)\in[\tau/n,2/n]$. But then the calls to the routine on Step~\ref{algo:gns:step:compare:2} will always return either a value (since $t$ is ``comparable'' to all $x\in L_{\eps/2}$ -- \ie, has probability within a factor $2/\tau$ of them) or $\textsf{High}$ (possible for 
those $s_j$      
that have weight greater than $2/n$), unless $s_j$ has mass $0$ (that is, is not in the support). Therefore, the fraction of points marked as support is exactly $\eta$, which by the foregoing discussion is at least $1-3\eps/4$: the algorithm returns \accept at Step~\ref{algo:gns:step:accept}.
\item Conversely, if $\omega \leq (1-\eps)n$, there will be
a fraction $1-\eta  > 3\eps/4$ of
the $s_j$        
having mass 0. However, no matter what $t$ is it will still be in the support and therefore have $\D(t) \geq \tau/n$: for
these $s_j$,    
the call to \textsc{Compare} on Step~\ref{algo:gns:step:compare:2} can thus only return $\textsf{Low}$. This means that there can only be less than $(1-\frac{3}{4}\eps)m$ points marked ``support'' among
the $s_j$,       
and hence that the algorithm will return \reject as it should.
\end{enumerate}
Overall, the inner loop of the algorithm thus only fails with probability
at most $1/12+1/6+1/12=1/3$,
(where the 3 events contributing to the union bound are (i) when $\eta$ fails to be a good estimate,
(ii) when the calls to \textsc{Compare} fail to yield results as
claimed, and (iii)  
when no $t_j$ hits $L_{\eps/2}$ in case \textsf{(a)}).  
Repeating independently $\log(1/\delta)$ times and taking the
majority vote boosts the probability of success to $1-\delta$.

\paragraph{Query complexity.} The sample complexity comes from the $k^2$ calls on Step~\ref{algo:gns:step:compare}  (each costing $\bigO{\log k}$ queries) and the $m$ calls on Step~\ref{algo:gns:step:compare:2} (each costing $\bigO{\frac{1}{\tau}\log m}$ queries). By the setting of $m$ and because of the $\log(1/\delta)$ repetitions, this results in an overall query complexity $\bigO{\left(\frac{1}{\tau^2}\log\frac{1}{\tau}+\frac{1}{\tau\eps^2}\log \frac{1}{\eps}\right)\log \frac{1}{\delta}}$.\qed

\subsection{Proof of~\expref{Lemma}{lemma:getnonsupport}}

As described in~\expref{Algorithm}{algo:get:non:support}, the subroutine is fairly simple: using its knowledge of an upperbound on the support size, it takes enough uniformly distributed samples to have (with high probability) at least one falling outside the support. Then, it uses the conditional oracle to ``order'' these samples according to their probability mass, and returns the lightest of them -- \ie, one with zero probability mass.

\begin{algorithm}[h]
  \begin{algorithmic}[1]
    \Require \COND access to $\D$; upper bound $m$ on $\supp{\D}$, probability of failure $\delta$
    \Ensure Returns $r\in\domain$ such that, with probability at least $1-\delta$,  $r\notin\supp{\D}$
    \State Set $k\eqdef \clg{\log \frac{2}{\delta}\log^{-1}\frac{n}{m}}$.
    \State Draw independently $k$ points $s_1,\dots, s_{k}\sim \uniformOn{\domain}$
    \ForAll{ $1\leq i < j\leq k$ }
      \State Call $\textsc{Compare}(\{s_i\},\{s_j\},\eta=\frac{1}{2},K=2,\frac{\delta}{2k^2})$ to get a $2$-approx. $\rho$ of $\frac{\D(s_j)}{\D(s_i)}$, \textsf{High} or \textsf{Low}. \label{algo:gns:step:compare}       \If{ $\textsc{Compare}$ returned \textsf{High} or a value $\rho$ }
        \State Record $s_i \preceq s_j$
      \Else
        \State Record $s_j \prec s_j$
      \EndIf
    \EndFor
    \State \Return $\arg\min_{\preceq} \{s_1,\dots,s_k\}$ \Comment{Return (any) minimal element for $\preceq$.}
  \end{algorithmic}
  \caption{\label{algo:get:non:support}$\textsc{GetNonSupport}_\D(m, \delta)$}
\end{algorithm}

\paragraph{Correctness.} It is straightforward to see that provided at least one of
the $s_j$       
falls outside the support and that all calls to  $\textsc{Compare}$ behave as expected, then the procedure returns one of the ``lightest''
elements $s_j$,   
\ie\ a non-support element. By a union bound, the latter holds with probability at least $1-\delta/2$; as for the former, since $m$ is by assumption an upper bound on the support size it holds with probability at least $1-(m/n)^k \geq 1-\delta/2$ (from our setting of $k$). Overall, the procedure's output is correct with probability at least $1-\delta$, as claimed.

\paragraph{Query complexity.} The query complexity of $\textsc{GetNonSupport}$ is due to the $\binom{k}{2}$ calls to $\textsc{Compare}$, and is therefore $\bigO{k^2 \log \frac{k}{\delta}}$ because of our setting for $\eta$ and $K$ (which is in turn $\tildeO{\log^2\frac{1}{\delta}\log^{-2}\frac{n}{m}}$). (In our case, we shall eventually take $m=(1-\eps/2)n$ and $\delta=1/10$, thus getting $k=\bigO{1/\eps}$ and a query complexity of $\tildeO{1/\eps^2}$.) \qed

\subsection{Proof of~\expref{Lemma}{lemma:isatmostsupportsize}}

Our final subroutine, described in~\expref{Algorithm}{algo:iamss}, essentially derives from the following observation: a random set $S$ of size (approximately) $\sigma$ obtained by including independently each element of the domain with probability $1/\sigma$ will intersect the support on $\omega/\sigma$ points on expectation. What we \emph{can} test given our reference point $r\notin\supp{\D}$, however, is only whether $S\cap\supp{\D} = \emptyset$. But this is enough, as by repeating sufficiently many times (taking a random $S$ and testing whether it intersects the support at all) we can distinguish between the two cases we are interested in. Indeed, the expected fraction of times $S$ includes a support element in either cases is known to the algorithm and differs by roughly $\bigOmega{\eps}$, so $\bigO{1/\eps^2}$ repetitions are enough to tell the two cases apart.

\begin{algorithm}[H]
  \begin{algorithmic}[1]
    \Require \COND access to $\D$; size $\sigma \geq 2$, non-support element $r$, accuracy \eps, probability of failure $\delta$
    \Ensure Returns, with probability at least $1-\delta$, \yes if $\sigma \leq \abs{\supp{\D}}$ and  \no if $\sigma > (1+\eps)\abs{\supp{\D}}$.
    \State Set $\alpha \gets \left(1-\frac{1}{\sigma}\right)^\sigma\in[\frac{1}{4},e^{-1}]$, $\tau\gets \alpha(\alpha^{-\frac{\eps}{2}}-1) = \bigTheta{\eps}$.
    \State Repeat the following $\bigO{\log(1/\delta)}$ times and return the majority vote.
    \Loop
      \For{$m=\bigO{\frac{1}{\tau^2}}$ times} \label{algo:iamss:atomic}
          \State Draw a subset $S\subseteq\domain$ by including independently each $x\in\domain$ with probability $1/\sigma$.
          \State Draw $x\sim \D_S$.
          \State\label{algo:iatss:step:compare} Call $\textsc{Compare}(\{x\},\{r\},\eta=\frac{1}{2},K=1,\frac{1}{100m})$
           \Comment{\textsf{Low} if $S\cap\supp{\D}\neq \emptyset$; $\rho\in[\frac{1}{2},2)$ o.w.}
          \State\label{algo:iamss:record:yes} Record \yes if $\textsc{Compare}$ returned \textsf{Low}, \no otherwise.
        \EndFor
        \State\label{algo:iamss:thresholding} \Return \yes if at least $m\left(\alpha+\frac{\tau}{2}\right)$ ``\yes'''s were recorded, \no otherwise. \Comment{Thresholding.}
    \EndLoop
  \end{algorithmic}
  \caption{\label{algo:iamss}$\textsc{IsAtMostSupportSize}_\D(\sigma, r, \eps, \delta)$}
\end{algorithm}

\paragraph{Correctness.} We condition on all calls to $\textsc{Compare}$ being correct: by a union bound, this overall happens with probability at least $99/100$. We shall consider the two cases $\sigma \leq \omega$ and $\sigma > (1+\eps)\omega$, and focus on the difference of probability $p$ of recording \yes on Step~\ref{algo:iamss:record:yes} between the two, in any fixed of the $m$ iterations. In both cases, note $p$ is exactly $(1-1/\sigma)^\omega$.
\begin{itemize}
  \item If $\sigma \leq \omega$, then we have $p \leq \left(1-\frac{1}{\sigma}\right)^\sigma = \alpha$.
  \item If $\sigma > (1+\eps)\omega$, then $p > \left(1-\frac{1}{\sigma}\right)^{\sigma/(1+\eps)} > \left(1-\frac{1}{\sigma}\right)^{\sigma(1-\eps/2)} = \alpha^{1-\eps/2}$.
\end{itemize}
As $\alpha\in[\frac{1}{4},e^{-1}]$, the difference between the two is $\tau=\alpha(\alpha^{-\eps/2} -1 )=\bigTheta{\eps}$. Thus, repeating the atomic test of Step~\ref{algo:iamss:atomic} $\bigO{1/\tau^2}$ before thresholding at Step~\ref{algo:iamss:thresholding} yields the right answer with constant probability, then brought to $1-\delta$ by the outer repeating and majority vote.

\paragraph{Query complexity.} Each call to $\textsc{Compare}$ at Step~\ref{algo:iatss:step:compare} costs $\bigO{\log m}$ queries, and is in total repeated $\bigO{m\log(1/\delta}$ times. By the setting of $m$ and $\tau$, the overall query complexity is therefore $\bigO{\frac{1}{\eps^2}\log\frac{1}{\eps}\log\frac{1}{\delta}}$. \qed

\subsection{A Non-Adaptive Upper Bound}\label{ssec:nonadaptive:support:size:ub}
In this section, we sketch how similar -- yet less involved -- ideas can be used to derive a non-adaptive upper bound for support-size estimation.
For simplicity, we describe the algorithm for $2$-approximation: adapting it to general $(1+\eps)$-approximation is straightforward.

The high-level idea is to perform a simple binary search (instead of the double exponential search from the preceding section) to identify the greatest lower bound on the support size of the form $k=2^j$. For each guess $k\in\{2,4,8\,\dots,n\}$, we pick uniformly at random a set $S\subseteq[n]$ of cardinality $k$, and check whether $\D_S$ is uniform using the non-adaptive tester of Chakraborty et al.~\cite[Theorem 4.1.2]{CFGM:13}. If $\D_S$ is found to be uniform for all values of $k$, we return $n$ as our estimate (as the distribution is close to uniform on $[n]$); otherwise, we return $n/k$, for the smallest $k$ on which $\D_S$ was found to be far from uniform. Indeed, $\D_S$ can only be far from uniform if $S$ contains points from the support of $\D$, which intuitively only happens if $n/k=\bigOmega{1}$.\medskip 

\noindent To be more precise, the algorithm proceeds as follows, where $\tau >0$ is an absolute constant.
\begin{algorithmic}
  \ForAll{$k\in\{2,4,\dots,n\}$}
    \State Set a counter $c_k\gets 0$.
    \For{$m=\bigO{\log\log n}$ times}
      \State Pick uniformly at random a set $S\subseteq[n]$ of $k$ elements.
      \State Test (non-adaptively) uniformity of $\D_S$ on $S$, with the tester of~\cite{CFGM:13}.
      \If{the tester rejects}
        increment $c_k$.
      \EndIf
    \EndFor
    \If{$c_k > \tau\cdot m$} \Return $\tilde{\omega}\gets \frac{n}{k}$.
    \EndIf
  \EndFor
  \State \Return $\tilde{\omega}\gets n$.
\end{algorithmic}
The query complexity is easily seen to be $\poly{\log n}$, from the $\tildeO{\log n}$ calls to the $\poly(\log n)$ tester of~\cite[Theorem 4.1.2]{CFGM:13}. As for correctness, it follows from the fact that for any set $S$ with mass $\D(S)>0$ which contains at least an $\eta$ fraction of points outside the support, it holds that $\D_S$ is $\eta$-far from $\uniformOn{S}$.

\paragraph{Acknowledgments.}
Cl\'ement Canonne would like to thank Dana Ron and Rocco Servedio for the many helpful discussions and remarks that influenced the lower bound construction of~\expref{Section}{sec:lb-equiv}. The authors are grateful to L\'aszl\'o Babai, Robert Krauthgamer, and the anonymous reviewers for their helpful and detailed comments.

\nocite{ValiantValiant:10ub}

\bibliographystyle{alpha} 
\bibliography{v014a988}

\end{document}